\documentclass[12pt]{article} 
\usepackage{cite}
\usepackage{amsmath,amssymb,amsthm,graphicx}

\numberwithin{equation}{section}

\begin{document}

\newcommand{\beq}{\begin{equation}}
\newcommand{\eeq}{\end{equation}}
\newcommand{\bea}{\begin{eqnarray}}
\newcommand{\eea}{\end{eqnarray}}
\newcommand{\sinhq}{\sinh}
\newcommand{\coshq}{\cosh}
\newcommand{\einde}{$\ \ \ \Box$ \vspace{5mm}}
\newcommand{\De}{\Delta}
\newcommand{\de}{\delta}
\newcommand{\Z}{{\mathbb Z}}
\newcommand{\N}{{\mathbb N}}
\newcommand{\C}{{\mathbb C}}
\newcommand{\Cs}{{\mathbb C}^{*}}
\newcommand{\R}{{\mathbb R}}
\newcommand{\Q}{{\mathbb Q}}
\newcommand{\T}{{\mathbb T}}
\newcommand{\cW}{{\cal W}}
\newcommand{\cJ}{{\cal J}}
\newcommand{\cE}{{\cal E}}
\newcommand{\cA}{{\cal A}}
\newcommand{\cR}{{\cal R}}
\newcommand{\cP}{{\cal P}}
\newcommand{\cM}{{\cal M}}
\newcommand{\cN}{{\cal N}}
\newcommand{\cI}{{\cal I}}
\newcommand{\cB}{{\cal B}}
\newcommand{\cD}{{\cal D}}
\newcommand{\cC}{{\cal C}}
\newcommand{\cL}{{\cal L}}
\newcommand{\cF}{{\cal F}}
\newcommand{\cH}{{\cal H}}
\newcommand{\cS}{{\cal S}}
\newcommand{\cT}{{\cal T}}
\newcommand{\cU}{{\cal U}}
\newcommand{\cQ}{{\cal Q}}
\newcommand{\cV}{{\cal V}}
\newcommand{\cK}{{\cal K}}
\newcommand{\diag}{{\rm diag}}

\title{ Tzitzeica solitons vs. \\relativistic Calogero-Moser 3-body clusters }
\author{J. J. C. Nimmo\thanks{Department of Mathematics, University of Glasgow,
Glasgow, G12 8QW, UK} \and 
S. N. M. Ruijsenaars
\thanks{Department of Applied Mathematics, University of Leeds, Leeds LS2 9JT, UK}
\thanks{Department of Mathematical Sciences, Loughborough University, Loughborough LE11 3TU, UK} 
}
\date{}
\maketitle

\begin{abstract}
\noindent
We establish a connection between the hyperbolic relativistic Calogero-Moser systems and a class of soliton solutions to the Tzitzeica equation (aka the Dodd-Bullough-Zhiber-Shabat-Mikhailov equation). In the $6N$-dimensional phase space~$\Omega$ of the relativistic systems with $2N$ particles and $N$ antiparticles, there exists a $2N$-dimensional Poincar\'e-invariant submanifold $\Omega_P$ corresponding to $N$ free particles and $N$ bound particle-antiparticle pairs in their ground state. The Tzitzeica $N$-soliton tau-functions under consideration are real-valued, and obtained via the dual Lax matrix evaluated in points of $\Omega_P$. This correspondence leads to a picture of the soliton as a cluster of two particles and one antiparticle in their lowest internal energy state.
\end{abstract}

\tableofcontents

\newtheorem{lem}{Lemma}[section]
\newtheorem{theor}[lem]{Theorem}
\newtheorem{cor}[lem]{Corollary}
\newtheorem{prop}[lem]{Proposition}



\renewcommand{\theequation}{\thesection.\arabic{equation}}

\setcounter{equation}{0}
\section{Introduction}

The equation
\beq\label{Tz}
\Psi_{uv}=e^{\Psi}-e^{-2\Psi}
\eeq
has a curious history. It first arose a century ago in the work of the Rumanian mathematician Tzitzeica~\cite{tzit1,tzit2}. He arrived at it from the viewpoint of the geometry of surfaces, obtaining an associated linear representation and a B\"acklund transformation.

For many decades after Tzitzeica's work, the equation (\ref{Tz}) was not studied, two papers by Jonas~\cite{jo1,jo2} being a notable exception. Thirty years ago, it was reintroduced within the area of soliton theory, independently by Dodd and Bullough~\cite{dobu} and Zhiber and Shabat~\cite{zhsh}, cf. also Mikhailov's paper~\cite{mikh}. In this setting, (\ref{Tz}) is viewed as an integrable relativistic theory for a field $\Psi(t,y)$ in two space-time dimensions, written in terms of light cone (characteristic) coordinates,
\beq\label{ty}
t=u-v,\ \ \ \ y=u+v.
\eeq
 Accordingly, the PDE (\ref{Tz}) is known under various names, and has been studied from several perspectives, including geometry~\cite{tzit1,tzit2,jo1,jo2}, classical soliton theory~\cite{dobu,zhsh,mikh,schi1,weis,atni,muco,scro,schi2,brez,chsh,kash,bozy,hita,
wazw,miro}, and quantum soliton theory~\cite{smir,frmu,acer,falu,bull,asfe}. Moreover, it has shown up within the context of gas dynamics~\cite{gaff,rosc}.

The principal aim of this paper is to tie in a class of soliton solutions to the Tzitzeica equation (\ref{Tz}) with integrable particle dynamics of relativistic Calogero-Moser type. (A survey covering both relativistic and nonrelativistic Calogero-Moser systems can be found in~\cite{scmt}.) The intimate relation of the latter integrable particle systems to soliton solutions of various evolution equations (including the sine-Gordon, Toda lattice, KdV and modified KdV equations) was already revealed in the paper in which they were introduced~\cite{rusc}, and was elaborated on in~\cite{aa2}. Later on, the list of equations whose soliton solutions are connected to the relativistic Calogero-Moser systems was considerably enlarged~\cite{fdss, KP, ade}. In all of these cases, the $N$ solitons correspond to $N$ point particles. 

The novelty of the present soliton-particle correspondence is that the Tzitzeica $N$-soliton solutions at issue correspond to an integrable reduction of the $3N$-body relativistic Calogero-Moser dynamics. Physically speaking, a Tzitzeica soliton may be viewed as a lowest energy bound state of three Calogero-Moser `quarks', one of which has negative charge, whereas the other two have positive charge.

A crucial ingredient for establishing the correspondence is the relation between an extensive class of 2D Toda solitons and the relativistic Calogero-Moser systems, already studied in~\cite{KP}. Indeed, the relation can be combined with the link between the Tzitzeica equation and the 2D Toda equation. The latter link has been known for quite a while, and we proceed to sketch it in a form that suits our later requirements.

Assume that $\phi_n$ is a solution to the 2D Toda equation in the form~\cite{mikh2}
\beq\label{To}
\phi_{n,uv}=\exp(\phi_n-\phi_{n-1})-\exp(\phi_{n+1}-\phi_n),\ \ \ \ n\in\Z,
\eeq
which has the symmetry property
\beq\label{sym}\phi_{-n}=-\phi_{n},
\eeq
and which is moreover 3-periodic, i.e.,
\beq\label{per}
\phi_{n+3}=\phi_n.
\eeq
Then one has in particular
\beq
\phi_0=0,\ \ \ \phi_2=-\phi_1,
\eeq
so that 
\beq
\Psi = \phi_1
\eeq
satisfies (\ref{Tz}). Conversely, a solution $\Psi$ to (\ref{Tz}) yields a solution $\phi_n$ to (\ref{To}) satisfying (\ref{sym}) and (\ref{per}) when one sets
\beq
\phi_{3k}=0,\ \ \ \phi_{1+3k}=-\phi_{2+3k}=\Psi,\ \ \ \ k\in\Z.
\eeq

The point is now that there exist soliton solutions to (\ref{To}) that can be made to satisfy the extra requirements (\ref{sym})--(\ref{per}), hence yielding soliton solutions to (\ref{Tz}). The relevant 2D Toda solitons are those found by the Kyoto school~\cite{daka,jimi}. These solitons also formed the starting point for~\cite{KP}. They are most easily expressed in tau-function form, the relation of $\tau_n$ to $\phi_n$ being given by
\beq\label{phitau}
\phi_n=\ln (\tau_{n+1}/\tau_{n}),\ \ \ \ n\in\Z.
\eeq
In terms of $\tau_n$, the evolution equation becomes
\beq\label{Totau}
\partial_u \partial_v \ln \tau_n =1- \tau_{n-1}\tau_{n+1}/\tau_n^2,
\eeq
and the extra features (\ref{sym}) and (\ref{per}) amount to
\beq\label{symt}
\tau_{-n+1}=\tau_n,
\eeq
and
\beq\label{pert}
\tau_{n+3}=\tau_n.
\eeq

As we show in Section~2, one can make special parameter choices in the 2D Toda $2N$-soliton solutions $\tau_n(u,v)$ so that they satisfy (\ref{symt})--(\ref{pert}). The function
\beq\label{Psi}
\Psi =\ln (\tau_2/\tau_1)
\eeq
then satisfies (\ref{Tz}), and can be viewed as a Tzitzeica $N$-soliton solution. (In Lie algebraic terms, the successive requirements (\ref{symt}) and (\ref{pert}) amount to reductions $A_{\infty}\to B_{\infty}\to A_2^{(2)}$.)

More specifically, the 2D Toda $2N$-solitons of Section~2 are of the form
\beq\label{tauTo}
\tau_n(u,v)=\det ({\bf 1}_{2N}+D(n,u,v)C),
\eeq
where the dependence of the $2N\times 2N$  (Cauchy type) matrix $C$ and diagonal matrix $D$ on the parameters $a,b,\xi^0\in\C^{2N}$ is suppressed.  To satisfy the $B_{\infty}$ restriction (\ref{symt}) and to prepare for the 3-periodicity restriction (\ref{pert}), these $6N$ parameters are expressed in terms of $2N$ parameters $\phi,\theta\in\C^N$ and a coupling parameter $c$. We then show that the tau-functions have period $l$ for $c$ equal to $\pi/l$, so that the Tzitzeica restrictions are satisfied for $c=\pi/3$.

In Section~3 we make a further parameter change, trading $\phi_1,\dots ,\phi_N$ for `positions' $q_1,\ldots,q_N$. This reparametrization ensures in particular that the summand involving all exponentials has coefficient 1. Restricting attention to the parameter set
\beq\label{cPT}
\cP= \{ (q,\theta)\in\R^{2N}\mid \theta_N<\cdots<\theta_1\},
\eeq
the tau-functions take their simplest and most natural form. In particular, for parameters in $\cP$ the tau-functions are real-valued. Furthermore, their space-time dependence is such that the $q_j$'s and $\theta_j$'s can be interpreted as relativistic positions and rapidities. Last but not least, it is in this form that the $c=\pi/3 $ Tzitzeica $N$-soliton tau-functions can be most easily compared to the tau-functions arising in the framework of the $3N$-body relativistic Calogero-Moser systems.

Section~4 is devoted to this comparison. The choice of regime for the Calogero-Moser systems is the same as for almost all other soliton equations. Specifically, the regime is the hyperbolic one, with the Poincar\'e group generators given by
\beq\label{Poi}
H=\frac{M_0}{2}(S_{+}+S_{-}),\ \ \ \ P=\frac{M_0}{2}(S_{+}-S_{-}),\ \ \ \ B=\sum_{i=1}^{N_{+}}x_i^{+}
+\sum_{j=1}^{N_{-}}x_j^{-},
\eeq
\beq\label{Spm}
S_{\pm}=\sum_{1\le i\le N_{+}}\exp(\pm p_i^{+})V_i^{+} +
\sum_{1\le j\le N_{-}}\exp(\pm p_j^{-})V_j^{-},
\eeq
\beq\label{V+}
(V_i^{+})^2    =  \prod_{1\le k\le N_{+}, k\ne i}\left( 1+\frac{\sin^2c}{\sinh^2(x_i^{+}-x_k^{+})/2}\right)
  \prod_{1\le j\le N_{-}} \left( 1-\frac{\sin^2c}{\cosh^2(x_i^{+}-x_j^{-})/2}\right),
\eeq
\beq\label{V-}
(V_j^{-})^2    =    \prod_{1\le l\le N_{-}, l\ne j}\left( 1+\frac{\sin^2c}{\sinh^2(x_j^{-}-x_l^{-})/2}\right)
 \prod_{1\le i\le N_{+}} \left( 1-\frac{\sin^2c}{\cosh^2(x_j^{-}-x_i^{+})/2}\right).
\eeq
Here, the generalized positions and momenta vary over the phase space
\beq
\Omega =\{ (x^{+},x^{-},p^{+},p^{-})\in\R^{2(N_{+}+N_{-})}\mid x_{N_{+}}^{+}<\cdots <x_1^{+},x_{N_{-}}^{-}<\cdots <x_1^{-}\},
\eeq
equipped with the symplectic form
\beq
\omega =\sum_{1\le i\le N_{+}}dx_i^{+}\wedge dp_i^{+}+
\sum_{1\le j\le N_{-}}dx_j^{-}\wedge dp_j^{-}.
\eeq
(Only the Landau-Lifshitz solitons of~\cite{daji} involve the more general elliptic regime~\cite{fdss}.) The coupling~$c$ that is needed, however, differs from the value $\pi/2$ relevant for most soliton equations (which correspond to $A_1^{(1)}$). Indeed, as already mentioned, the connection of the space-time dynamics (\ref{Poi})--(\ref{V-}) to the Tzitzeica tau-functions arises for the `$A_2$'-value
\beq
c=\pi/3.
\eeq

In order to clarify further restrictions and to describe our main result in some detail, a few more ingredients need to be introduced. First, the Poisson commuting `light-cone' Hamiltonians $S_{+}$ and $S_{-}$ can be expressed in terms of an $(N_{+}+N_{-}) \times (N_{+}+N_{-})$ Lax matrix on $\Omega$ via 
\beq
S_{+}={\rm Tr}(L),\ \ \ S_{-}={\rm Tr}(L^{-1}).
\eeq
This matrix has a product structure $D\cM D$, with a diagonal matrix $D$ and a matrix $\cM$ that arises by suitable substitutions in Cauchy's matrix $\cC(x,y)$, i.e., a matrix with elements 
$(x_i-y_j)^{-1}$, whose determinant is given by Cauchy's identity
\beq\label{Cauchy}
|\cC(x,y)|=\frac{\prod_{i<j}(x_i-x_j)(y_j-y_i)}{\prod_{i,j}(x_i-y_j)}.
\eeq
Hence one can explicitly determine the symmetric functions of $L$ and its inverse, and all of these functions Poisson commute.

Specifically, the Lax matrix
\beq
L=D\cM D
\eeq
 is given by
\bea
D  &  =  &  \diag(\exp(p_1^{+}/2)(V_1^{+})^{1/2},\ldots,\exp(p_{N_{+}}^{+}/2)(V_{N_{+}}^{+})^{1/2},
\nonumber \\
  &  &  \exp(p_1^{-}/2)(V_1^{-})^{1/2},\ldots,\exp(p_{N_{-}}^{-}/2)(V_{N_{-}}^{-})^{1/2}),
\eea
\bea
\cM_{ik}  &  =  &  \frac{i\sin c}{\sinh((x_i^{+}-x_k^{+})/2+ic)},
\nonumber\\
\cM_{N_{+}+j,N_{+}+l}  &  =  &  \frac{i\sin c}{\sinh((x_j^{-}-x_l^{-})/2+ic)},
\nonumber\\
\cM_{N_{+}+j,k}  &  =  &  \frac{\sin c}{\cosh((x_j^{-}-x_k^{+})/2+ic)},
\nonumber\\
\cM_{i,N_{+}+l}  &  =  &  -\overline{\cM}_{N_{+}+l,i},
\eea
where $i,k\in \{ 1,\ldots,N_{+}\}$ and $j,l\in\{ 1,\ldots, N_{-}\}$. It is clear from this that $L$ is self-adjoint when $N_{+}$ or $N_{-}$ vanish. Physically speaking, this is the case in which only particles or antiparticles are present. For $N_{+}N_{-}>0$ however, $L$ is not self-adjoint. Instead, $L$ is $J$-self-adjoint, that is, the adjoint $L^{*}$ equals $JLJ$, where
\beq
J=\diag({\bf 1}_{N_{+}},-{\bf 1}_{N_{-}}).
\eeq

When one takes the interparticle distances to infinity, it becomes clear that there is a subset of $\Omega$ on which $L$ is diagonalizable with real eigenvalues. (Of course, for the one-charge case this is true on all of $\Omega$.) However, already for the simplest non-self-adjoint case $N_{+}=N_{-}=1$, complex-conjugate eigenvalues are also present, reflecting the existence of bound states of a particle and antiparticle. More precisely, choosing from now on
\beq\label{cran}
c\in (0,\pi/2),
\eeq
these bound states are encoded by eigenvalues 
\beq
\exp (\hat{\theta}_{+}),\exp(\hat{\theta}_{-}),\ \ \ \ \hat{\theta}_{\pm}=\theta \pm i\kappa,\ \ \ \kappa\in(0,c].
\eeq
Moreover, the case of an eigenvalue pair on the boundary of this sector corresponds to
phase space points
\beq\label{deadb}
x_{+}=x_{-}=x,\ \ p_{+}=p_{-}=p.
\eeq
For $p=0$ this yields a 1-parameter set with minimal energy $H=2M_0\cos c$, cf.~(\ref{Poi})--(\ref{V-}), and, more generally, no oscillation takes place for the 2-parameter set (\ref{deadb}).

The special case just discussed is the only one that can be explicitly understood in an elementary way. Already for $N_{+}+N_{-}=3$ (the simplest case relevant for this paper) a complete account involves considerable analysis. The general case has been elucidated in great detail in~\cite{aa2}, and we need to make extensive use of Sections~2B--6B of that paper, which pertain to the $c$-range (\ref{cran}). The action-angle map constructed there makes it possible to understand all of the Poisson commuting dynamics at once, including their soliton type long-time asymptotics (conservation of momenta and factorized position shifts).

It is beyond our scope to recapitulate even the case $N_{+}+N_{-}=3$, but we do mention the key starting point for the analysis performed in~\cite{aa2}. This is because it clarifies why the diagonal matrix
\beq\label{A1}
A(x)= \diag (\exp(x_1^{+}),\ldots,\exp(x_{N_{+}}^{+}),-\exp(x_1^{-}),\ldots,-\exp(x_{N_{-}}^{-})),
\eeq
is of pivotal importance. This matrix is referred to as the dual Lax matrix, and its relation to $L$ is encoded in the commutation relation
\beq\label{comrel}
 i\cot (c) (LA-AL)=LA+AL-2e\otimes e.
 \eeq
 (This relation is readily verified; we do not need the dyadic in the sequel.)  
 
 One need only inspect (\ref{comrel}) to see that $A$ and $L$ play symmetric roles. In the one-charge case it is possible to diagonalize the self-adjoint matrix $L$ in such a way that $A$ takes essentially the same form in terms of suitable variables (which are just the action-angle variables). This self-duality is no longer present for $N_{+}N_{-}>0$, however. Indeed, $A$ is then still self-adjoint, whereas $L$ is not. Even so, (\ref{comrel}) can again be used to great advantage in the construction of the action-angle map, as detailed in~\cite{aa2}.

We are now prepared to specialize the above to the arena with which the present paper is concerned. First of all, we choose
\beq\label{NN}
N_{+}=2N,\ \ \ N_{-}=N.
\eeq
 This choice ensures that there exists a $2N$-dimensional 
Poincar\'e-invariant submanifold 
\beq\label{TT}
\Omega_P\simeq \cP,
\eeq
with $\cP$ given by (\ref{cPT}), of the $6N$-dimensional phase space $\langle \Omega,\omega\rangle$. This submanifold will be described in detail in Section~4, and we shall presently add a brief qualitative description.

Consider now the solution 
\beq\label{jsol}
\exp(tH-yP)Q,\ \ \ Q=(x,p)\in\Omega,
\eeq
to the joint Hamilton equations for $H$ and $P$. Such a joint solution exists and is given by \eqref{jsol}, since the $H$-flow and $P$-flow commute. Using (\ref{ty}) and (\ref{Poi}), we can rewrite \eqref{jsol} as
\beq\label{uvf}
\exp(M_0(uS_{-}-vS_{+}))Q=(x(u,v),p(u,v)).
\eeq
Next, specializing to
\beq\label{Tspec}
M_0=\sqrt{3},\ \ \ c=\pi/3,
\eeq
we introduce
\beq\label{tauL}
\tau_n(u,v)= \det ({\bf 1}_{3N}+\exp(i\pi (1-2n)/3)A(x(u,v))).
\eeq
For a given point $Q$ in the phase space $\Omega$, this yields a function depending on $n\in\Z$ and $(u,v)\in \R^2$.

We are finally in the position to state the principal result of this paper: When the tau-function (\ref{tauL}) is restricted to the $2N$-dimensional Poincar\'e-invariant subspace $\Omega_P$ of $\Omega$, it coincides with the Tzitzeica tau-function (\ref{tauTo}) evaluated on $\cP$ (\ref{cPT}).
Our demonstration of this equality uses in particular special cases of the fusion identities obtained in~\cite{KP}. (For completeness we add the proof of these specializations in Appendix~A.) The reduction in the matrix size from $3N\times 3N$ to $2N\times 2N$ hinges on all points in $\Omega_P$ yielding pairs of complex-conjugate eigenvalues $\exp(\theta_j\pm i\pi/3)$ for the Lax matrix $L$ on the boundary of the allowed angular sector. For each of these pairs there is an additional eigenvalue $\exp(\theta_j)$, so that the spectrum of $L$ on $\Omega_P$ involves only $N$ degrees of freedom $\theta_1,\ldots ,\theta_N\in \R$ (and not the $3N$ of the full phase space). These variables may be viewed as action variables, and there are $N$ canonically conjugate `angle' variables $q_1,\ldots,q_N\in \R$. As a consequence, $\Omega_P$ can be identified with $\cP$, cf.~(\ref{TT}) and (\ref{cPT}).

A better understanding of how $\Omega_P$ arises as a submanifold of the $6N$-dimensional phase space can only be achieved by invoking a great many details concerning the action-angle map constructed in~\cite{aa2}, which are summarized in Section~4. At this point we only add a few qualitative remarks, so as to render these details more accessible. 

First, the above coordinates are not quite action-angle coordinates. Rather, the precise definition of $\Omega_P$ involves the harmonic oscillator transform. This transform is an extension of the action-angle transform, which takes into account that $\Omega$ contains an open dense subset that is the union of submanifolds for which the angles vary over $\R^{3N-l}\times \T^l$, where $\T \simeq (-\pi,\pi]$ and $l$ takes all values in $ \{ 0,1,\ldots,N\}$. From a physical point of view, these submanifolds can be regarded as the subsets of $\Omega$ on which $l$ particle-antiparticle bound states are present. Now when the bound state internal actions converge to their minima, the $l$-torus collapses to lower-dimensional tori in precisely the same way as for a harmonic oscillator Hamiltonian $\sum_{j=1}^l (p_j^2+x_j^2)$, which motivates the terminology.

In these terms, $\Omega_P$ amounts to the subset that arises from the submanifold with $N$ bound states by taking the torus $\T^N$ to a point, so that each of the pairs is in its ground state; moreover, each of the remaining particles has action and angle variables that are paired with those of the bound states, in such a way that the asymptotic space-time dependence of $\Psi$ is that of $N$ 3-body clusters moving apart, each of the clusters staying together as in the $N=1$ case.

In Section~5 we present a close-up of the latter case. For $N=1$, various questions can be rather easily answered, and this special case is also useful as an illustration of the general case.  In particular, we obtain some explicit information on the 2-dimensional space $\Omega_P$ for $c\in(0,\pi/2)$, and study the Tzitzeica 1-soliton solution corresponding to $c=\pi/3$.

In Section~6 we compare the Tzitzeica solitons under consideration to the ones obtained by Kaptsov and Shanko~\cite{kash} and to the solitons arising by the above two-step reduction (\ref{symt})--(\ref{pert}) from a class of 2D Toda solitons constructed via Darboux transformations. At face value, these two types of Tzitzeica solitons seem different from the ones obtained from the Kyoto solitons in Section~2. As we show, however, the latter are  a subclass of the former.

We have added Section~6 primarily because it yields an affirmative answer to the two natural equality questions at issue, but as a bonus it yields a new insight on the Tzitzeica tau-function $\tau_0$ (given both by (\ref{tauL}) and (\ref{tauTo})): It is in fact positive, and equal to the square of a simpler tau-function occurring in the Kaptsov-Shanko work~\cite{kash}.

In Section~7 we consider further aspects of our results, in the form of several remarks. Specifically, we show that the space-time translation generators restricted to $\Omega_P$ give rise to a reduced integrable system, we introduce and comment on space-time trajectories for the Tzitzeica solitons,  isolate the difference with the setup of~\cite{KP}, and comment on an eventual quantum analog of the correspondence between 3-body clusters and solitons. 

A final remark concerns the solitons of the Demoulin system of equations~\cite{demo,rosc2}. As it turns out, their tau-function form is obtained by taking $c=\pi/6$ in Sections~2 and~ 3. It is a challenging question whether they can also be tied in with the relativistic Calogero-Moser systems. 

We have relegated a proof of the fusion identities used in Section~4 to Appendix~A. In Appendix~B we collect some auxiliary results concerning pfaffians, which we need in Section~6. Finally, Appendix~C contains a sketch of the Darboux type construction of explicit tau-function solutions to the 2D Toda equation (\ref{Totau}). 



\renewcommand{\theequation}{\thesection.\arabic{equation}}

\setcounter{equation}{0}
\section{The soliton reduction 2D Toda $\to$ Tzitzeica}

Our starting point is the 2D Toda $M$-soliton solutions introduced by the Kyoto school~\cite{daka,jimi} in the form
\beq \label{taun}
\tau_n =\sum_{\mu_1,\ldots,\mu_M=0,1}\,
\prod_{1\le j\le
M}\exp(\mu_j \xi_{j,n})\cdot
\prod_{1\le j<k\le M}f_{jk}^{\mu_j \mu_k},
\eeq
\beq \label{eB}
f_{jk}=\frac{(a_j-a_k)(b_j-b_k)}{(a_j-b_k)(b_j-a_k)},\ \
\ \ \ \ j,k=1,\ldots ,M,
\eeq
 \beq \label{xito}
\xi_{j,n}= \xi_{j}^0 + n \ln (a_j/b_j)
+i\sum_{l =1}^{\infty}(t_{l
,+}(a_j^l -b_j^l )+t_{l ,-}
(a_j^{-l }-b_j^{-l })).
  \eeq
  The sequence of `times' $t_{l,\pm}$ encodes the hierarchy of commuting flows. Restricting attention to the 2D Toda equation (\ref{Totau}), one should set
  \beq\label{trest}
  t_{l,+}=t_{l,-}=0,\ \ \ l>1,\ \ \ t_{1,+}=-v,\ \ \ t_{1,-}=u,
  \eeq
  to obtain a solution.
  
We proceed in several steps to show how a suitable specialization of the $3M$ parameters $a,b$ and $\xi^{0}$ gives rise to a tau-function with the `$B_{\infty}$'-symmetry (\ref{symt}). This also involves the choice
  \beq
  M=2N,\ \ \ N\in\N^{*},
  \eeq
which will be in force from now on.

\begin{prop}\label{prop:cauchy}
Define a Cauchy type matrix $C$ with elements
\beq\label{C}
C_{jk}=\frac{a_j-b_j}{a_j-b_k},\ \ \ \ j,k=1,\ldots, M,
\eeq
and a diagonal matrix
\beq\label{D}
D_n=\diag (\exp(\xi_{1,n}),\ldots, \exp(\xi_{M,n})).
\eeq
Then $\tau_n$ 
may be rewritten as
\beq\label{taun det}
\tau_n=|{\bf 1}_M+D_nC|.
\eeq
\end{prop}  
\begin{proof}
	  Using Cauchy's identity \eqref{Cauchy}, the principal minor expansion of the rhs of \eqref{taun det} yields
	  \beq\label{taupm}
	  \tau_n=\sum_{l=0}^M\sum_{I\subset\{ 1,\ldots,M\}, |I|=l}\, \prod_{j\in I}\exp(\xi_{j,n})\cdot \prod_{j,k\in I, j<k}f_{jk},
	  \eeq
	which amounts to (\ref{taun}).
\end{proof}    

Next, we specialize the parameters so as to achieve the $B_\infty$ reduction. 
Specifically, we choose	
 \beq\label{ba}
 b_{M-j+1}=-a_j,\ \ \ \ j=1,\ldots,M,
 \eeq
 \beq\label{xiz}
 \exp(\xi_j^{0})=\exp(\varphi_j)a_{M-j+1}(a_j+a_{M-j+1})^{-1},\ \ \ \ j=1,\ldots,N,
 \eeq
\beq\label{xiz2}
 \exp(\xi_{M-j+1}^{0})=-\exp(\varphi_{j})a_{j}(a_j+a_{M-j+1})^{-1},\ \ \ \ j=1,\ldots,N.
 \eeq
\begin{prop}\label{prop:B}
With \eqref{ba}--\eqref{xiz2} in effect, $\tau_n$ has the $B_\infty$ symmetry \eqref{symt}.
\end{prop}	  
\begin{proof}
 Setting
	 \beq\label{chi}
	 \chi_j= \exp((\varphi _j+\psi_j)/2),\ \ \ \ j=1,\ldots,M,
	 \eeq
	 with
	 \beq\label{phi}
	 \varphi _{M-j+1}= \varphi _j,\ \ \ \ j=1,\ldots,N,
	 \eeq
	 \beq\label{psi}
	  \psi_j= -iv(a_j+a_{M-j+1})+iu(a_j^{-1}+a_{M-j+1}^{-1}),\ \ \ \ j=1,\ldots,M,
	  \eeq
we clearly have
	  \beq
	  \chi_{M-j+1}=\chi_j,\ \ \ \ j=1,\ldots,M.
	  \eeq
Also, the reparametrized matrix $D_nC$ is the product of the matrix
	  \beq\label{cJ}
	  \cJ= \diag({\bf 1}_N,-{\bf 1}_N),
	  \eeq
and the matrix with $jk$-element
	  \beq\label{repar}
	  \chi_j^2\left(-\frac{a_j}{a_{M-j+1}}\right)^n \frac{a_{M-j+1}}{a_j+a_{M-k+1}},\ \ \ \ j,k=1,\ldots,M.
	  \eeq
  
After a similarity transformation $\tau_n$ may be rewritten as
  \beq
 \label{tre}
 \tau_n=|{\bf 1}_M+\cD \cA_n \cD \cJ|,
 \eeq
 where
 \beq\label{calD}
 \cD= \diag (\chi_1,\ldots,\chi_N,\chi_N,\ldots,\chi_1),
 \eeq
 \beq\label{calA}
 \cA_{n,jk}= \frac{(-a_j)^n(a_{M-k+1})^{-n+1}}{a_j+a_{M-k+1}},\ \ \ j,k=1,\ldots,M.
 \eeq
 From this definition of $\cA_n$ we readily infer that
 \beq
 \cA_{-n+1,jk}=-\cA_{n,M-k+1,M-j+1}.
 \eeq
 Hence we obtain
 \beq
 \cA_{-n+1}=-\cR_M\cA_n^t\cR_M,
 \eeq
 where the superscript denotes the transpose, and $\cR_M$ denotes the $M\times M$ reversal permutation matrix. Now we have
 \beq
 \cR_M\cD\cR_M=\cD,\ \ \ \ \cR_M\cJ\cR_M=-\cJ.
 \eeq
 From these formulae we deduce
 \bea
 \tau_{-n+1}  &  =  &  |{\bf 1}_M-\cD\cR_M\cA_n^t\cR_M\cD\cJ|=|{\bf 1}_M+\cR_M\cD\cA_n^t\cD\cJ\cR_M|
  \nonumber \\
    &  =  & | {\bf 1}_M+\cD\cA_n^t\cD\cJ|=|{\bf 1}_M+\cJ\cD\cA_n\cD|
    \nonumber\\
      &  =  &  \tau_n,
      \eea
as required.
\end{proof}

Finally, we show how by a further specialization of parameters, periodic reductions of the solitons of the $B_\infty$ Toda lattice arise.
We trade the $2N$ parameters $a_1,\ldots,a_M$ for $N$ parameters $\theta_1,\ldots,\theta_N$ and another parameter $c$, which corresponds to the coupling parameter in (\ref{Poi})--(\ref{V-}):
\beq\label{as1}
  a_j=\exp(\theta_j-ic),\ \ \ j=1,\ldots,N,
  \eeq
   \beq\label{as2}
  a_{M-j+1}=-\exp(\theta_{j}+ic),\ \ \ j=1,\ldots,N.
  \eeq                                 

\begin{prop}\label{prop:periodic}
With \eqref{as1}--\eqref{as2} in force, we have an implication
	\beq\label{impl}
	 c=\pi/l \Rightarrow \tau_{n+l}=\tau_n,\ \ \ \ l=1,2,\ldots.
	\eeq
In particular, the tau-function is 3-periodic for $c=\pi/3$, and hence is a tau-function form of the Tzitzeica $N$-soliton solutions.	
\end{prop}              

\begin{proof}                                                   
Using also (\ref{xiz})--(\ref{psi}), we obtain
\beq\label{xin1}
\exp(\xi_{j,n})=\frac{e^{ic}}{2i\sin c}\chi_j^2e^{-2inc},
\eeq
  \beq\label{xin2}
\exp(\xi_{M-j+1,n})=\frac{e^{-ic}}{2i\sin c}\chi_j^2e^{2inc},
\eeq  
\beq\label{chirep}
\chi_j^2=\exp(\varphi_j -2\sin(c)[ve^{\theta_j}+ue^{-\theta_j}]),
  \eeq
where $j=1,\ldots,N$.                               
From this \eqref{impl} is obvious.
\end{proof}
                  


\renewcommand{\theequation}{\thesection.\arabic{equation}}

\setcounter{equation}{0}

\section{Explicit form of the tau-functions}
In this section we first make the soliton formulae described in the last section more explicit, exemplifying them for the simplest case.  A suitable reparametrization of the $\phi_j$'s in (\ref{xiz})--(\ref{xiz2}) then yields a sum formula for the arbitrary-$c$ tau-function that is not only compact and informative, but which can also be tied in with the relativistic Calogero-Moser tau-function (\ref{tauL}) for the special $c$-value $\pi/3$.

First, from \eqref{taun}--\eqref{trest}, the 2-soliton solution of the 2D Toda lattice is
\begin{equation}\label{eq:Toda2}
\tau_n=1+e^{\xi_{1,n}}+e^{\xi_{2,n}}+\frac{(a_1-a_2)(b_1-b_2)}{(a_1-b_2)(b_1-a_2)}e^{\xi_{1,n}+\xi_{2,n}}.	
\end{equation}
After the reduction to the $B_\infty$ Toda lattice described in Proposition~\ref{prop:B}, \eqref{xito}--\eqref{trest}  yield
\begin{equation}\label{eq:Bxi}
	\exp(\xi_{j,n})=\frac{(-1)^na_j^na_{M-j+1}^{1-n}}{a_j+a_{M-j+1}}
	e^{\varphi_j(u,v)},\quad	
	\exp(\xi_{M-j+1,n})=\frac{(-1)^{n+1}a_j^{1-n}a_{M-j+1}^{n}}{a_j+a_{M-j+1}}
	e^{\varphi_j(u,v)},	
\end{equation}
 where                                                       
\begin{equation}\label{eq:Bphi}
	\varphi_j(u,v)=\varphi_j+i[(a_j^{-1}+a_{M-j+1}^{-1})u-(a_j+a_{M-j+1})v],
\end{equation}
and $j=1,\dots,N$, and \eqref{eB} becomes
\begin{equation}\label{eq:Bf}
f_{jk}=\frac{(a_j-a_k)(a_{M-j+1}-a_{M-k+1})}{(a_j+a_{M-k+1})(a_{M-j+1}+a_{k})},
\end{equation}                                                                 
where $j,k=1,\dots,M$. The 2-soliton solution of 2D Toda \eqref{eq:Toda2} becomes the 1-soliton solution of the $B_\infty$ Toda equation, given by
\begin{equation}\label{eq:B1}
\tau_n=1+\frac{(-1)^n(a_1^na_2^{1-n}-a_1^{1-n}a_2^n)}{a_1+a_2}
e^{\varphi_1(u,v)}+	
\frac{a_1a_2(a_1-a_2)^2}{(a_1+a_2)^4}e^{2\varphi_1(u,v)}.	
\end{equation}
                               
The result of the specialization \eqref{as1}--\eqref{as2} 
 is that \eqref{eq:Bxi}--\eqref{eq:Bphi}  become
\begin{equation}\label{eq:pf}
	\exp(\xi_{j,n})=\frac{e^{(1-2n)ic}}{2i\sin c}
	e^{\varphi_j(u,v)},\quad	
	\exp(\xi_{M-j+1,n})=\frac{e^{-(1-2n)ic}}{2i\sin c}
	e^{\varphi_j(u,v)},	
\end{equation}                                                       
\begin{equation}\label{eq:pphi}
	\varphi_j(u,v)=\varphi_j-2\sin(c)[ue^{-\theta_j}+ve^{\theta_j}],
\end{equation}  
where $j=1,\dots,N$.
Using \eqref{as1}--\eqref{as2} in \eqref{eq:Bf}, we finally obtain the following key expressions for the functions $f_{jk}$:
\beq\label{B1}
f_{jk}=\frac{\sinh^2((\theta_j-\theta_k)/2)}{\sinh^2((\theta_j-\theta_k)/2)+\sin^2c},\quad
f_{j,M-k+1}=\frac{\cosh^2((\theta_j-\theta_k)/2)-\sin^2c}{\cosh^2((\theta_j-\theta_k)/2)},
\eeq                                                                                      
for $j,k=1,\dots,N$ and
\beq\label{Bsymm}
f_{jk}=f_{M-j+1,M-k+1},
\eeq
for $j,k=1,\dots,M$. 
Furthermore, the specialization of the 1-soliton solution of $B_\infty$ Toda \eqref{eq:B1} gives 
\begin{equation}\label{eq:p1}
\tau_n=1+\frac{\cos((2n-1)c)}{i\sin c}
e^{\varphi_1(u,v)}-	
\frac{\cos^2c}{4\sin^2c}e^{2\varphi_1(u,v)}.	
\end{equation}         

So far in this section we have simply restated various
formulae in the last section in more explicit form. Now we make a final reparametrization of the phase constants $\varphi_j$.
Specifically, choosing from now on $(q,\theta)\in \cP$ (cf.~(\ref{cPT})), we set
\beq\label{vphi}
\exp(\varphi_j)=2i\sin(c)\exp(q_j/2)F_j(\theta),\ \ \ j=1,\ldots,N,
  \eeq
\beq\label{Fj}
F_j(\theta)=\prod_{k=1,\ldots,M, k\ne j}f_{jk}^{-1/2},\ \ \ j=1,\ldots,N.
\eeq
Substituting this in \eqref{eq:pf}--\eqref{eq:pphi},  we obtain   
\beq\label{xif1}
\exp(\xi_{j,n})=\exp((1-2n)ic+q_j(u,v)/2)F_j(\theta),
\eeq
\beq\label{xif2}
\exp(\xi_{M-j+1,n})=\exp((2n-1)ic+q_j(u,v)/2)F_j(\theta),
\eeq
\beq\label{quv}
q_j(u,v)=q_j-4\sin(c)[ve^{\theta_j}+ue^{-\theta_j}],
\eeq 
where $j=1,\ldots,N$.
Note from \eqref{B1} that the functions $f_{jk}$ are positive and hence so are $F_1,\ldots,F_N$.

We have now arrived at the final form of the tau-function. In the proof of the next proposition we obtain an explicit sum formula to show it is real-valued.
\begin{prop}\label{prop:tz sol}
The tau-function $\tau_n$ given by \eqref{taupm} with the factors defined by \eqref{B1}--\eqref{Bsymm} and \eqref{Fj}--\eqref{quv} is real-valued. In particular, for $c=\pi/3$ it yields a real-valued $N$-soliton solution of the Tzitzeica equation.   
\end{prop}                          
\begin{proof}   
It remains to prove that $\tau_n$ is real. It is clear from (\ref{taupm}) that $\tau_n$ is equal to a sum of terms $T_S$ associated with all subsets $S$ of $\{1,\dots,M\}$. 
To obtain an explicit expression for $T_S$,
we write the subsets as
\beq\label{defS}
S=\{ i_1,\ldots,i_s,M-j_b+1,\ldots,M-j_1+1\},
\eeq
where
\beq\label{ic}
1\le i_1<\cdots <i_s\le N,\ \ 1\le j_1<\cdots <j_b\le N,\ \ s,b\in \{ 0,1,\ldots,N\}.
\eeq
Then we deduce from \eqref{taupm} that the contribution to $\tau_n$ corresponding to $S$ is given by
\beq\label{TS}
T_S=\exp\left((b-s)(2n-1)ic+\sum_{\sigma=1}^s q_{i_{\sigma}}(u,v)/2
+\sum_{\beta=1}^b q_{j_{\beta}}(u,v)/2\right)F_S(\theta),
\eeq
where 
\beq\label{FSt}
F_S(\theta)=\prod_{j\in S,k\notin S}f_{jk}^{-1/2}.
\eeq                  

Next, defining the reversal permutation $r$ of $\{1,2,\dots,M\}$ by $r(j)=M-j+1$, it  follows from \eqref{Bsymm} that
\begin{equation}
F_{r(S)}=F_S,	
\end{equation}
and then from \eqref{TS} that  
\begin{equation}
T_{r(S)}=\overline{T_S},	
\end{equation}         
where the over bar denotes the complex conjugate. Hence, if $r(S)$ is equal to $S$ then $T_S$ is real, and otherwise $T_S+T_{r(S)}$ is real. Thus $\tau_n$ can be written as a sum of real terms. 
\end{proof}

To illustrate this result, we write the 1-soliton tau-function in the explicit form described in the above proof. We have
\begin{align}
	&T_\emptyset=1,\quad T_{\{1,2\}}=\exp(q_{1}(u,v))F_{\{1,2\}},\\
	&T_{\{1\}}=\exp\left((1-2n)ic+q_{1}(u,v)/2\right)F_{\{1\}},\\
	&T_{\{2\}}=\exp\left((2n-1)ic+q_{1}(u,v)/2\right)F_{\{2\}},
\end{align}
where (recall \eqref{B1} and $c\in(0,\pi/2)$)

\begin{equation}
F_{\{2\}}=
F_{\{1\}}=f^{-1/2}_{12}=1/\cos c,	\ \ \ 
F_{\{1,2\}}=1.
\end{equation}
Clearly, the terms $T_\emptyset$ and $T_{\{1,2\}}$ are real, and so is the sum
\begin{equation}
T_{\{1\}}+T_{\{2\}}=\frac{2\cos((2n-1)c)}{\cos c}\exp(q_1(u,v)/2).
\end{equation}                                                                         
The 1-soliton tau-function is given by
\begin{equation}\label{eq:pp1}
\tau_n=1+\frac{2\cos((2n-1)c)}{\cos c}e^{q_1(u,v)/2}+e^{q_1(u,v)}.	
\end{equation}         
In particular, for the Tzitzeica case $c=\pi/3$
we have  
\begin{equation}\label{eq:T1}
\tau_0=\tau_1=1+2e^{q_1(u,v)/2}+e^{q_1(u,v)}=(1+e^{q_1(u,v)/2})^2,	
\end{equation}         
and
\begin{equation}\label{eq:T1'}
\tau_2=1-4e^{q_1(u,v)/2}+e^{q_1(u,v)}.	
\end{equation}                   

Finally, we sketch how to obtain the 2-soliton tau-function in a similar way. In this case $M$ equals $4$ and for the 16 subsets  of $\{1,2,3,4\}$ there are 10 real-valued combinations 
\begin{equation}
	\begin{split}
	&T_\emptyset,\quad T_{\{1,4\}},\quad T_{\{2,3\}},\quad T_{\{1,2,3,4\}},\quad T_{\{1\}}+T_{\{4\}},\quad T_{\{2\}}+T_{\{3\}},\\
	&T_{\{1,2\}}+T_{\{3,4\}},\quad
	T_{\{1,3\}}+T_{\{2,4\}},\quad
	T_{\{1,2,3\}}+T_{\{2,3,4\}},\quad
	T_{\{1,2,4\}}+T_{\{1,3,4\}},
	\end{split}                  
\end{equation}
whose sum gives the 2-soliton tau-function. 



\renewcommand{\theequation}{\thesection.\arabic{equation}}

\setcounter{equation}{0}
\section{Tzitzeica solitons vs.~relativistic Calogero-Moser dynamics}

In the Introduction we have already sketched in general terms how the $2N$-dimensional Poincar\'e-invariant subset $\Omega_P$ of the phase space $\langle \Omega,\omega\rangle$ arises. In order to fill in the details of this qualitative picture, we begin by specifying the subspace of $\Omega$ on which the maximal number $N$ of bound states is present, choosing the internal actions at first larger than their minimum, so that there are $N$ angles varying over $\T^N$. Moreover, until further notice we work with $c\in (0,\pi/2)$, since this eases the notation and adds insight on the Tzitzeica case $c=\pi/3$. With this starting point understood, the variables $\hat{q},\hat{\theta}$ in the equations (4.20)--(4.24) of~\cite{aa2} should be specialized as follows.

First, as already mentioned (cf.~(\ref{NN})), the particle and antiparticle numbers $N_{+}$ and $N_{-}$ of~\cite{aa2} must be chosen equal to $2N$ and $N$ (so that the number $N$ used in~\cite{aa2} becomes $3N$), and the parameters of~\cite{aa2} should be specialized as
\beq\label{parc}
\mu=\beta=1,\ \ \ \tau =g/2=c.
\eeq
 Second, the bound state number~$l$ of~\cite{aa2} should be taken equal to $N$. Thus the numbers $k_{+}$ and $k_{-}$ in (4.20)--(4.24) are equal to $N$ and 0, resp., and we obtain $2N$ particle variables
\beq\label{pv}
\hat{q}_1,\ldots,\hat{q}_N\in\R,\ \ \ \hat{\theta}_N<\cdots <\hat{\theta}_1,
\eeq
and $4N$ bound state variables
\beq\label{bs1}
\hat{q}_{N+1},\ldots,\hat{q}_{2N},\hat{\theta}_{N+1},\ldots,\hat{\theta}_{2N}\in\C,\ \ \ \Im (\hat{\theta}_{N+j})\in (0,c),\ \ \ \hat{\theta}_{N+j}\ne \hat{\theta}_{N+k},
\eeq\label{bs2}
\beq\label{cc}
\hat{q}_{2N+j}=\overline{\hat{q}}_{N+j}+\frac{i\pi}{2}(1-(-)^N),\ \ \ \hat{\theta}_{2N+j}=\overline{\hat{\theta}}_{N+j},
\eeq
where $j,k=1,\ldots,N$ and $k\ne j$.

One advantage of these variables is that they give rise to compact expressions for the symmetric functions of the dual Lax matrix $A$. Specifically, they are given by
\beq\label{Sl}
S_l(A)=\sum_{|I|=l}\exp\left(\sum_{i\in I}\hat{q}_i\right)p_I(\hat{\theta}),\ \ \ \ l=0,\ldots,3N,
\eeq
with
\beq\label{pI}
p_I(\hat{\theta})=\prod_{i\in I,j\notin I}\frac{[\sinh^2((\hat{\theta}_i-\hat{\theta}_j)/2)+\sin^2c]^{1/2}}{\sinh((\hat{\theta}_{\min (i,j)}-\hat{\theta}_{\max (i,j)})/2)},
\eeq
cf.~(5.68) in~\cite{aa2} with parameters (\ref{parc}).

The action-angle variables are now given by
\beq\label{xjs}
x_j^s=\hat{q}_j,\ \ \ p_j^s=\hat{\theta}_j,
\eeq
\beq\label{xj}
x_j=\Re(\hat{q}_{N+j}),\ \ \ p_j=2\Re(\hat{\theta}_{N+j}),
\eeq
\beq
\gamma_j=\Im(\hat{q}_{N+j})+(N+j-1)\pi,\ \ \ \de_j=-2\Im(\hat{\theta}_{N+j}),
\eeq
where $j=1,\ldots,N$ and $\gamma_j\in (-\pi,\pi]$ (mod $2\pi$), cf.~(4.26) in~\cite{aa2}. To handle the limits $\de_j\downarrow -2c$, for which the $\gamma_j$-torus reduces to a  point, we should switch from $\gamma_j,\de_j$ to the `harmonic oscillator variables' $u_j,v_j$, the minimum $\de_j=-2c$ corresponding to the origin $u_j=v_j=0$, cf.~Chapter~5 in~\cite{aa2}. We shall presently analyze the behavior of $S_l(A)$ under this limit. Before embarking on this rather technical issue, however, we clarify how the limit gives rise to a $2N$-dimensional submanifold that is invariant under the Poincar\'e (inhomogeneous Lorentz) group. 

To this end we use the explicit description of the commuting flows generated by Hamiltonians of the form
\beq\label{Hf}
H_h={\rm Tr}\, h(\ln L),
\eeq
with $h(z)$ an arbitrary entire function (cf.~(6.79) in~\cite{aa2}), which follows from (6.6) in~\cite{aa2}. It reads
\bea\label{flows}
\exp(tH_h)(x^s,p^s;x,p,0,0)  &  =   &  (x_1^s+th'(p_1^s),\ldots,x_N^s+th'(p_N^s),p^s;
x_1+t\Re(h'(p_1/2-ic)),
\nonumber\\
  &  & \ldots,x_N+t\Re(h'(p_N/2-ic)),p,0,0).
\eea
The time and space translation generators $H$ and $P$ (cf.~(\ref{Poi})), for which $h(z)$ equals $M_0\cosh(z)$ and $M_0\sinh(z)$, resp., reduce to
\beq
H=M_0\sum_{j=1}^N(\coshq (p_j^s)+2\cos(c)\coshq (p_j/2)),\ \ \ 
P=M_0\sum_{j=1}^N(\sinhq (p_j^s)+2\cos(c)\sinhq (p_j/2)),
\eeq
on the submanifold at issue, yielding flows
\begin{align}\label{Hfl}                      
e^{tH/M_0}(x^s,p^s;x,p,0,0)  &  =   (x_1^s+t\sinhq (p_1^s),\ldots,x_N^s+t\sinhq (p_N^s),p^s;
x_1+t\cos(c)\sinhq (p_1/2),\nonumber\\
  \qquad &  \ldots,x_N+t\cos(c)\sinhq (p_N/2),p,0,0),\\
\label{Pfl}
e^{-yP/M_0}(x^s,p^s;x,p,0,0)  &  =   (x_1^s-y\coshq (p_1^s),\ldots,x_N^s-y\coshq (p_N^s),p^s;
x_1-y\cos(c)\coshq (p_1/2),\nonumber\\
  \qquad & \ldots,x_N-y\cos(c)\coshq (p_N/2),p,0,0).
\end{align}
The boost generator is given by (cf.~(\ref{Poi}))
\beq
B=\sum_{j=1}^N(x_j^s+2x_j),
\eeq
whence we have
\beq
e^{\alpha B}(x^s,p^s;x,p,0,0)=(x^s,p_1^s+\alpha,\ldots,p_N^s+\alpha;x,p_1+2\alpha,\ldots,p_N+2\alpha,0,0).
\eeq

From these formulae it is obvious that the $2N$-dimensional submanifold $\Omega_P$ obtained by requiring
\beq\label{xpr}
x_j=\cos(c)x_j^s,\ \ \ p_j=2p_j^s,\ \ \ j=1,\ldots,N,
\eeq
is Poincar\'e invariant. On $\Omega_P$ we have
\beq
H=p(c)M_0\sum_{j=1}^N\coshq p_j^s,\ \ \ P=p(c)M_0\sum_{j=1}^N\sinhq p_j^s,\ \ \ B=p(c)\sum_{j=1}^Nx_j^s,
\eeq
where the prefactor $p(c)$ is given by
\beq\label{pc}
p(c)=1+2\cos(c).
\eeq
Also, the symplectic form $\omega$ reduces to
\beq
\omega=p(c)\sum_{j=1}^N dx_j^s\wedge dp_j^s.
\eeq
Thus, when we reparametrize $\Omega_P$ with variables
\beq\label{qth}
q_j=p(c)x_j^s,\ \ \ \theta_j=p_j^s,\ \ \ j=1,\ldots,N,
\eeq
then $q,\theta$ are canonical coordinates. Moreover the space-time flow reads
\beq\label{stf}
e^{tH-yP}(q,\theta)=(q_1+p(c)M_0(t\sinhq \theta_1-y\coshq \theta_1),\ldots,
q_N+p(c)M_0(t\sinhq \theta_N-y\coshq \theta_N),\theta),
\eeq
or, equivalently,
\beq
e^{uS_{-}-vS_{+}}(q,\theta)=(q_1(u,v),\ldots,
q_N(u,v),\theta),
\eeq
where
\beq\label{quv2}
q_j(u,v)=q_j-p(c)M_0(ue^{-\theta_j}+ve^{\theta_j}),\ \ \ \ j=1,\ldots,N.
\eeq
With (\ref{Tspec}) in force, this coincides with the Tzitzeica soliton space-time dependence  (\ref{quv}). Of course, this state of affairs is still far from a proof that the tau-function defined at the end of Section~3 and the tau-function (\ref{tauL}) evaluated on $\Omega_P$ are equal.

In order to demonstrate this equality, we return to (\ref{Sl})--(\ref{pI}) and the variables $\hat{q},\hat{\theta}$. On $\Omega_P$ the latter specialize as
\beq\label{qhq}
\hat{q}_j=q_j/p(c),\ \ \ \Re(\hat{q}_{N+j})=\cos(c)q_j/p(c),\ \ \ j=1,\ldots,N,
\eeq
\beq\label{qhq0}  
\hat{\theta}_j=\theta_j,\ \ \ \hat{\theta}_{N+j}=\theta_j+ic,\ \ \ j=1,\ldots,N,
\eeq
and (\ref{cc}) still holds. More precisely, we have
\beq\label{qhq1}
\hat{\theta}_{2N+j}=\theta_j-ic,\ \ \ j=1,\ldots,N,
\eeq
but we can only infer
\beq\label{qs}
\hat{q}_{N+j}+\hat{q}_{2N+j}=2\cos(c)q_j/p(c)+\frac{i\pi}{2}(1-(-)^N),
\eeq
since the imaginary parts are undetermined in the collapsing torus limit.

We proceed to analyze this limit for (\ref{Sl})--(\ref{pI}). First, whenever $I$ contains the index $N+j$, but not the index $2N+j$, with $j=1,\ldots,N$, or vice versa, then it is clear from (\ref{pI}) that the limit of $p_I(\hat{\theta})$ vanishes. Thus we need only consider $I$ that either contain both indices or neither. As a consequence, we only encounter $\hat{q}_{N+j}$ and $\hat{q}_{2N+j}$ in the combination (\ref{qs}), which confirms that the limit is well defined.

Introducing the breather sets
\beq
B_j=\{ N+j,2N+j\},\ \ \ j=1,\ldots,N,
\eeq
we proceed to simplify the limit of $p_I(\hat{\theta})^2$ for $I$ of the form
\beq\label{If}
I=\{ i_1,\ldots, i_s\} \cup B_{j_1}\cup \cdots\cup B_{j_b},
\eeq
with (\ref{ic}) in effect. This boils down to a quite special case of the fusion procedure detailed on pp.~237--238 of~\cite{KP}. Thus we put
\beq\label{Jdef}
J=\{ i_1,\ldots, i_s,j_1+N,\ldots, j_b+N\} \subset \{ 1,\ldots, 2N\},
\eeq
and
\beq\label{eq:eta}
\eta_j=\theta_j,\ \ \eta_{j+N}=\theta_j,\ \ \ c_j=c,\ \ \ c_{j+N}=2c,\ \ \ j=1,\ldots,N,
\eeq
to obtain
\beq\label{limI}
\lim p_I(\hat{\theta})^2=\prod_{j\in J,k\notin J}\frac{\sinh^2((\eta_j-\eta_k)/2)+\sin^2((c_j+c_k)/2)}
{\sinh^2((\eta_j-\eta_k)/2)+\sin^2((c_j-c_k)/2)},
\eeq
cf.~(2.27) in~\cite{KP}. To render this step in our reasoning self-contained, we present a proof of (\ref{limI}) in Appendix~A.

Requiring (\ref{Tspec}) from now on, the following two cases arise for the product factor $p_{jk}$ on the rhs of (\ref{limI}):
\beq\label{p1}
j,k\le N\ {\rm or}\ j,k>N,\ {\rm with} \ j\ne k \Rightarrow p_{jk}=\frac{\sinh^2((\theta_j-\theta_k)/2)+3/4}
{\sinh^2((\theta_j-\theta_k)/2)},
\eeq
\beq\label{p2}
j\le N,k>N\ {\rm or}\ j>N,k\le N \Rightarrow p_{jk}=\frac{\sinh^2((\theta_j-\theta_k)/2)+1}
{\sinh^2((\theta_j-\theta_k)/2)+1/4}.
\eeq
Taking positive square roots of the $p_{jk}$, we deduce from (\ref{limI}) that we have
\beq\label{plim}
\lim p_I(\hat{\theta})=s_J \prod_{j\in J,k\notin J}p_{jk}^{1/2},
\eeq
for a certain sign $s_J$. We shall determine this sign later on, cf.~(\ref{sJ}).

Next, we focus on the principal minor expansion of the tau-function (\ref{tauL}), restricted to the submanifold of $\Omega$ with $N$ bound states present, with internal actions $\de_1,\ldots,\de_N$ in $(-2\pi/3,0)$. It reads
\bea
\tau_n(u,v)  &  =  &  \sum_{l=0}^{3N}\exp(i\pi l(1-2n)/3)S_l(A(x(u,v)))
\nonumber\\
  &  =  &   \sum_{l=0}^{3N}\exp(i\pi l(1-2n)/3)\sum_{|I|=l}\exp(\sum_{i\in I}\hat{q}_i(u,v))p_I(\hat{\theta}).
  \eea
Taking the actions to their minima $-2\pi/3$, we deduce from the above that the limiting tau-function is a sum of nonzero contributions $C_J$ for all $J$ of the form (\ref{Jdef}), with
\bea\label{CJ}
C_J  &  =  &  \exp(i\pi (s+2b)(1-2n)/3)\prod_{\sigma=1}^s \exp(q_{i_{\sigma}}(u,v)/2)
\nonumber\\
  &  &  \times \prod_{\beta=1}^b \exp\left(q_{j_{\beta}}(u,v)/2+\frac{i\pi}{2}\left(1-(-)^N\right)\right)\cdot
  s_J\prod_{j\in J,k\notin J}p_{jk}^{1/2}.
\eea

We are now prepared to state and prove the principal result of this paper.

\begin{theor} Let $c=\pi/3$. Then the relativistic Calogero-Moser tau-function (\ref{tauL}) restricted to the subspace $\Omega_P\simeq \cP$ is equal to the Tzitzeica tau-function (\ref{taupm}) with the factors given by (\ref{B1})--(\ref{Bsymm}) and (\ref{Fj})--(\ref{quv}) and with parameters in $\cP$.
\end{theor}
\begin{proof}
We need only show equality of the general contribution $C_J$ to the general term $T_S$ in (\ref{taupm}), cf.~(\ref{TS}).
First, we compare $f_{jk}$, given by (\ref{B1})--(\ref{Bsymm}) with $c=\pi/3$, to $p_{jk}$ given by (\ref{p1})--(\ref{p2}). From this we easily deduce
\beq
F_S(\theta)=\prod_{j\in J,k\notin J}p_{jk}^{1/2}.
\eeq
Next, we note that the factors involving $q_j(u,v)$ are in agreement. Hence the asserted equality comes down to an equality of the remaining numerical factors. Specifically, it remains to show
\beq\label{sid}
(-)^b\exp\left(\frac{i\pi b}{2}\left(1-(-)^N\right)\right)s_J=1.
\eeq

To this end we now calculate the sign $s_J$ in (\ref{plim}). We begin by noting that for $\hat{\theta}\in \C^{3N}$ given by (\ref{pv})--(\ref{bs2}), the product $p_I(\hat{\theta})$ (\ref{pI}) has a positive numerator. (Indeed, each radicand is either positive or has a nonzero imaginary part; in the latter case it is matched by a factor with the complex-conjugate radicand.) We are therefore reduced to analyzing the phase of
\beq
\Pi_I=\prod_{i\in I,j\notin I}s_{ij},\ \ \ s_{ij}=\sinh((\hat{\theta}_{\min(i,j)}-\hat{\theta}_{\max(i,j)})/2),
\eeq
for the special $\hat{\theta}$ under consideration, namely,
\beq\label{thsp}
\hat{\theta}=(\theta_1,\ldots,\theta_N,\theta_1+i\pi/3,\ldots,\theta_N+i\pi/3,\theta_1-i\pi/3,\ldots,\theta_N-i\pi/3),\ \ \ \ \theta_N<\cdots<\theta_1.
\eeq
We already know from (\ref{plim}) that this phase is just the sign $s_J$. Indeed, we have
\beq\label{thp}
\sinh((\theta_j-\theta_k-i\psi)/2)\sinh((\theta_j-\theta_k+i\psi)/2)>0,\ \ \ \psi=\pi/3,2\pi/3,
\eeq
which confirms that $\Pi_I$ is either positive or negative.

We continue to analyze the contributions to $s_J$ from indices in $I$ (\ref{If}), taking (\ref{thsp}) into account. First, we consider the contribution of $i_{\sigma}\in I$. Due to the ordering of the $\theta_j$ in (\ref{thsp}), any $j\notin I$ with $j\le N$ yields $s_{i_{\sigma}j}>0$. Also, any $B_k$ that is not a subset of $ I$ yields a contribution of the form (\ref{thp}) with $\psi=\pi/3$. Hence all indices $i_1,\ldots,i_s$ in $I$ yield positive signs.

Next, we study the contribution of $B_{j_{\beta}}\subset I$. For $i\in I$ with $i\le N$ it again follows from (\ref{thp}) that we get a positive sign. Now suppose $B_k$ is not a subset of $I$ and consider the pertinent product 
\beq
s_{j_{\beta}+N,k+N}s_{j_{\beta}+N,k+M}s_{j_{\beta}+M,k+N}s_{j_{\beta}+M,k+M}.
\eeq
Both for $k<j_{\beta}$ and for $k>j_{\beta}$ this product is easily seen to be negative. Therefore any pair $B_j,B_k$ with $B_j$ included in $I$ and $B_k$ not included in $I$ gives rise to a minus sign in $\Pi_I$. Now $I$ contains $b$ breather index sets, so there are $N-b$ breather sets not contained in $I$. Thus we finally deduce
\beq\label{sJ}
s_J=(-)^{b(N-b)}.
\eeq

With this explicit formula in hand, it is routine to verify (\ref{sid}). Consequently, our proof of the asserted tau-function equality is now complete.
\end{proof}


\renewcommand{\theequation}{\thesection.\arabic{equation}}

\setcounter{equation}{0}

\section{The $N=1$ case}
In this section we consider various aspects of the special case $N=1$, which yields an illuminating illustration of the above constructions and equalities. More specifically, we focus on features of the Tzitzeica 1-soliton solution and the space $\Omega_P$ for $c\in(0,\pi/2)$.

For $N=1$ it follows from (\ref{eq:T1})--(\ref{eq:T1'})  that we have
\beq\label{1sol}
\tau_0=\tau_1=1+2F+F^2,\ \ \tau_2=1-4F+F^2,\ \ \ F=\exp(q(u,v)/2),
\eeq
with
\beq
q(u,v)=q-2\sqrt{3}(ue^{-\theta}+ve^{\theta})=q+2\sqrt{3}(t\sinh(\theta)-y\cosh(\theta)),
\eeq
cf.~(\ref{ty}). Thus the two $\tau_2$-zeros for $F=2\pm \sqrt{3}$ yield two parallel space-time lines
\beq
y_{\pm}(t)=t\tanh(\theta)+\frac{1}{2\sqrt{3}\cosh(\theta)}(q-2\ln (2\pm \sqrt{3})),
\eeq
where $\Psi$ diverges, cf.~Fig.~1. 
In this figure and in Fig.~2, we plot for clarity $-\Psi$ rather than $\Psi$, and truncate $-\Psi$ at a finite cutoff value. The regions between the singularities, where $\tau_2$ is negative so that $\Psi=\ln (\tau_2/\tau_0)$ takes complex values, are approximated by the `plateaux' in these figures. 
  
\begin{figure}[h]
	\centering
	\includegraphics*[width=6in]{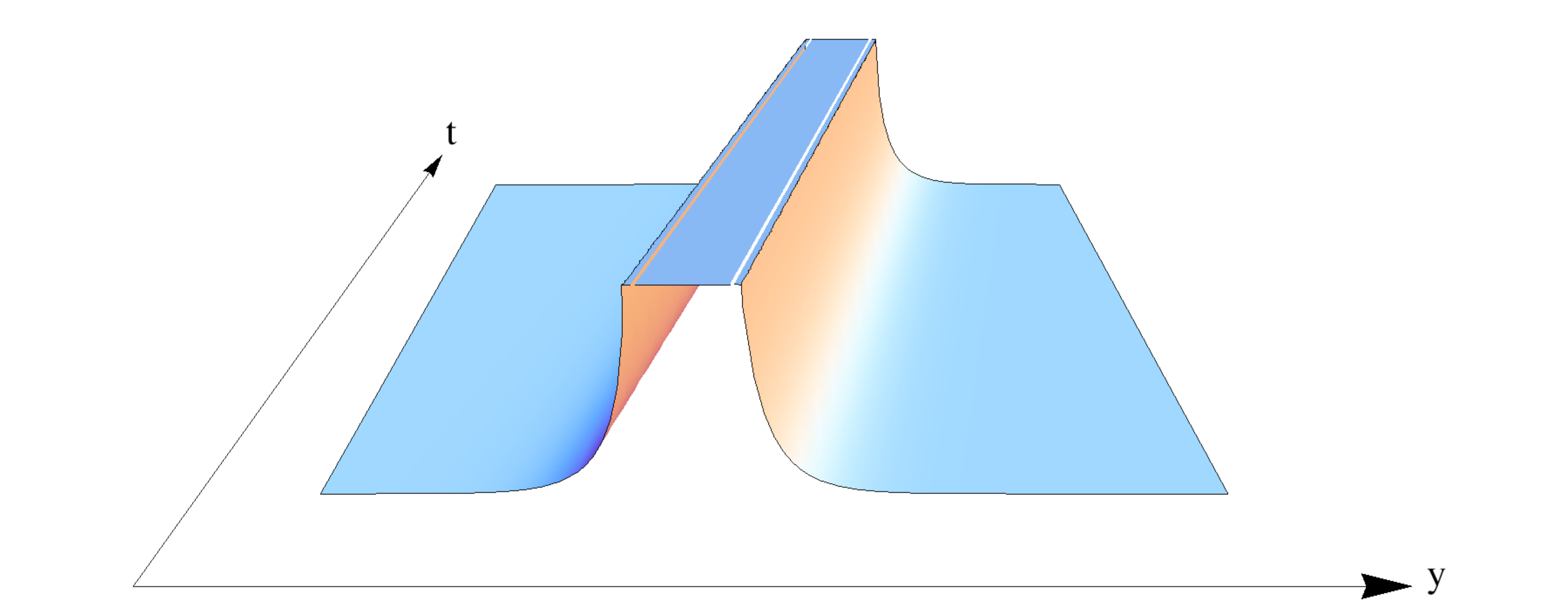}
	\caption{One-soliton solution. $-\Psi(t,y)$ is shown, the $y$-axis points left to right and the $t$-axis points into the page. The `plateau' approximates the region in which $\Psi(t,y)$ is complex-valued.}
\end{figure}
\
Recalling (\ref{tauL}), we see that the $\tau_2$-zeros correspond to $x_1^{+}(t,y)$ or $x_2^{+}(t,y)$ being zero. The space-time lines may therefore be viewed as the trajectories of the two particles in the 1-soliton cluster.

For  $N=1$  it might seem an easy matter to locate the 2-dimensional space $\Omega_P$ within the 6-dimensional phase space $\Omega$. In fact, however, it is not even easy to find the 1-dimensional submanifold of $\Omega_P$ consisting of $H$-equilibrium points, i.e.,
\beq
\Omega_E=\{ (q,\theta)\in\Omega_P\simeq \R^2\mid \theta =0\},
\eeq
in explicit form. (Note that for $c=\pi/3$ this yields the stationary 1-soliton tau-functions.) Of course, it is immediate from (\ref{Poi})--(\ref{V-}) that we need
\beq\label{peq}
p_1^{+}=p_2^{+}=p_1^{-}=0,
\eeq
for $H$ to have an equilibrium. Requiring this,
one expects from physical considerations to get an equilibrium point for
\beq\label{xeq}
x_1^{+}=-x_2^{+}=d>0,\ \ \ x_1^{-}=0,
\eeq
and a suitable $d$. We now confirm this for $c\in (0,\pi/2)$. Then we have from (\ref{Poi})--(\ref{V-})
\beq\label{Hequi}
\frac{H(d,-d,0,0,0,0)}{M_0}=2(f^{+}f^{-})^{1/2}+f^{-},\ \ \ f^{+}=
1+\frac{\sin^2c}{\sinh^2d},\ \ f^{-}=1-\frac{\sin^2c}{\cosh^2(d/2)}.
\eeq
Thus we have $H\to 3M_0$ for $d\to\infty$ and $H\to \infty$ for $d\to 0$. Since $H$ is equal to $M_0\sum_{j=1}^3 \cosh(\hat{\theta_j})$, it has an absolute minimum $M_0p(c)$ for $\hat{\theta}= (0,ic,-ic)$. Now we can still choose $x_1^s,x_1\in\R$ (cf.~(\ref{xjs})--(\ref{xj})), so $H$ has  a 2-parameter family of stable equilibria. It is straightforward to check that for $d$ equal to
\beq\label{de}
d_e= \cosh^{-1}(1+\cos c),
\eeq
we have
\beq
f^{+}=\frac{1+2\cos c}{(2+\cos c)\cos c},\ \ \ f^{-}=\frac{(1+2\cos c)\cos c}{2+\cos c}.
\eeq
Hence
the rhs of 
(\ref{Hequi}) equals $M_0p(c)$, so we do get a stable equilibrium for $d=d_e$. This implies that all of the points
\beq\label{eq1}
E(\sigma)=(d_e+\sigma, -d_e+\sigma, \sigma,0,0,0),\ \ \ \ \sigma\in\R,
\eeq
are also equilibria.

Since $H$ equals $M_0p(c)$ on $\Omega_E$, one might expect that the 1-parameter family of equilibria $E(\sigma)$ yields $\Omega_E$. In fact, however, only $E(0)$ belongs to $\Omega_E$. To be specific, we assert that it corresponds to $q=0$. To show this, we use (\ref{Sl}) with $N=1$ to find the symmetric functions of $A$ on $\Omega_E$. We determined the relevant limits below (\ref{qs}). In particular, (\ref{limI}) yields $p(c)^2$ for $I$ equal to $\{ 1\}$ or $\{ 2,3\}$, since $\theta_1=0$ on $\Omega_E$. Using also (\ref{qhq}) and (\ref{qs}), we readily obtain
\beq\label{S1}
S_1(A)=p(c)\exp(q/p(c)),
\eeq
\beq\label{S2}
S_2(A)=-p(c)\exp(2\cos(c)q/p(c)),
\eeq
\beq\label{S3}
S_3(A)=-\exp(q).
\eeq
For $q=0$ this implies that the spectrum of $A$ is given by
\beq
\sigma(A)=\{ 1+\cos(c)\pm (\cos^2(c)+2\cos(c))^{1/2},-1\}.
\eeq
The first two eigenvalues can be written as $\exp(\pm d_e)$, cf.~(\ref{de}). Thus the origin of $\Omega_E\simeq \R$ corresponds to generalized positions $x_1^{+}=d_e,x_2^{+}=-d_e, x_1^{-}=0$. Since $\Omega_E$ consists of $H$-equilibria, we also have $p_1^{+}=p_2^{+}=p_1^{-}=0$.  Hence we deduce $E(0)\in\Omega_E$ and $q=0$, as asserted.

On the other hand, the equilibrium points in $\Omega_P$ with $q\ne 0$ are harder to find explicitly. As already announced, they do not include the equilibria $E(\sigma)$ for $\sigma\ne 0$. (Indeed, $E(0)$ does belong to $\Omega_E$, as just shown. Now the $\sigma$-shift of the generalized positions corresponds to a shift of $x_1^s$ and $x_1$ by $\sigma$, so that the $\Omega_P$-condition  $x_1=\cos (c)x_1^s$ no longer holds true for $\sigma\ne 0$.) Rather, they are obtained by acting with the commuting flow $\exp(yP/M_0), y\in\R$, on the equilibrium
$E(0)$, yielding translated equilibria
\beq
T(y)=(x_1^{+}(y),x_2^{+}(y),x_1^{-}(y),0,0,0),\ \ \ y\in\R.
\eeq
In view of (\ref{S3}) and (\ref{stf}) this yields the sum rule
\beq
x_1^{+}(y)+x_2^{+}(y)+x_1^{-}(y)=p(c)y.
\eeq
Also, from Hamilton's equations for $P/M_0$ we have 
\beq
x_j^{+}(y)'>\cos(c),\ \ j=1,2,\ \ x_1^{-}(y)'\ge\cos^2(c),
\eeq
so that $x_1^{+}(y), x_2^{+}(y)$ and $x_1^{-}(y)$ are strictly increasing functions of $y$. 

Next, we deduce from (\ref{S1})--(\ref{S3}) that we have the reflection symmetry
\beq
x_1^{+}(-y)=-x_2^{+}(y), \ \ \ x_1^{-}(-y)=-x_1^{-}(y),
\eeq
so it remains to determine the functions for $y>0$. This is presumably possible, but we have not pursued this. However, the large-$y$ asymptotics can be established from the above. Specifically, setting
\beq
\epsilon =\exp (\cos c -1)y,
\eeq
we obtain for $\epsilon\to 0$
\beq
x_1^{+}(y)=y+\ln p(c) +O(\epsilon^2),
\eeq
\beq
x_2^{+}(y)= \cos(c) y-\frac{1}{2} \ln p(c) -\frac{1}{2}\sqrt{p(c)}(1-p(c)^{-2})\epsilon +O(\epsilon^2),
\eeq
\beq
x_1^{-}(y)= \cos(c) y-\frac{1}{2} \ln p(c) +\frac{1}{2}\sqrt{p(c)}(1-p(c)^{-2})\epsilon +O(\epsilon^2).
\eeq

A final observation of interest concerns the kinetic and potential energy density of the stationary 1-soliton solution, i.e., the functions
\beq
E_K(y)=(\partial_y\Psi)^2/2,\ \ \ \ E_P(y)=\exp(\Psi)+\exp(-2\Psi)/2-3/2,\ \ \ \Psi=\ln(\tau_2/\tau_0),
\eeq
obtained from (\ref{1sol}) for $\theta =0$. It is far from obvious, but true that these functions are equal. This equality can be verified directly by a straightforward, but quite tedious calculation.

A more conceptual derivation of this virial type identity will now be given. First, we note that any $t$-independent solution to (\ref{Tz}) satisfies the ODE
\beq\label{ode}
f_{yy}=e^f-e^{-2f}.
\eeq
Now from (\ref{1sol}) we obtain a solution to (\ref{ode}) of the form 
\beq
f(y)=g(\exp((q-2\sqrt{3}y)/2)),\ \ \ \ g(z)=\ln\left(\frac{1-4z+z^2}{1+2z+z^2}\right).
\eeq
Defining
\beq
z_{\pm}=2\pm \sqrt{3},
\eeq
it satisfies
\beq\label{eg+}
e^{g(z)}>0,\ \ \ z>z_{+},\ \ \ z<z_{-},
\eeq
\beq\label{eg-}
e^{g(z)}<0,\ \ \ z\in (z_{-},z_{+}),
\eeq
and
\beq\label{glim}
g(z),g'(z)\to 0,\ \ \ z\to\pm\infty,
\eeq
\beq\label{gsp}
\exp(g(1))=-1/2,\ \ \ g'(1)=0.
\eeq

Next, consider the Hamiltonians
\beq\label{Hpm}
H_{\pm}(x,p)=p^2/2-V_{\pm}(x),\ \ \ V_{\pm}(x)=\pm e^x+e^{-2x}/2-3/2.
\eeq
The potential $-V_{+}(x)$ has a maximum 0 at $x=0$ and yields a Newton equation
\beq
\ddot{x}=e^x-e^{-2x}.
\eeq
Comparing this to (\ref{ode}), (\ref{eg+}) and (\ref{glim}), we see that $g(z)$ for $z>z_{+}$ corresponds to the $E=0$ orbit with $x(t)<0$ coming from 0 for $t\to -\infty$, and $g(z)$ for $z<z_{-}$ to the $E=0$ orbit with $x(t)<0$ going to 0 for $t\to\infty$. By energy conservation, we have $\dot{x}^2/2=V_{+}(x)$, which implies $E_K(y)=E_P(y)$ for the $y$-intervals where $e^{f(y)}>0$.

It remains to show the identity for the $y$-interval where $e^{f(y)}<0$. Then we can compare (\ref{ode}) to the Newton equation
\beq
\ddot{x}=-e^x-e^{-2x},
\eeq
corresponding to $H_{-}$. The $E=0$ orbit stays to the left of the origin and has its turning point at $x=\ln (1/2)$, cf.~(\ref{Hpm}). Comparing this to (\ref{gsp}), we see that it corresponds to the real part of  $g(z)$ for $z\in(z_{-},z_{+})$. Hence $E_K(y)=E_P(y)$ now follows from $\dot{x}^2/2=V_{-}(x)$.


\renewcommand{\theequation}{\thesection.\arabic{equation}}

\setcounter{equation}{0}

\section{The Darboux and Kaptsov-Shanko solitons}

In this section we begin by showing that the solitons obtained by the $B_{\infty}$-reduction from the Kyoto 2D Toda solitons are equal to those obtained by a similar reduction from a seemingly different class of 2D Toda solitons. The latter are constructed via repeated Darboux transformations. As will transpire, this procedure yields a larger class of solutions, involving an arbitrary constant antisymmetric matric $\cC$.   In order to obtain equality to the $B_{\infty}$ solitons of Section~2, this matrix must be suitably specialized.

Our demonstration of equality leads to a new representation of the latter solitons. This representation can be exploited to show that $\tau_0$ equals the square of a simpler tau-function. This is because it readily leads to $\tau_0$ being the determinant of an antisymmetric $M\times M$ matrix $A$. Hence it follows that we have
\beq\label{tsq}
\tau_0=\tau^2,\ \ \ \ \tau= {\rm Pf}(A).
\eeq
The pfaffian can be explicitly evaluated, and when the Tzitzeica substitutions of Sections~2--3 are made in $A$, then the resulting $\tau$ is the one obtained by Kaptsov and Shanko in their study of the Tzitzeica equation~\cite{kash}. 

We proceed with the details. We start from the tau-function $\tau_n$ (\ref{tre}), with $\cJ$, $\cD$ and $\cA$ given by (\ref{cJ}), (\ref{calD}) and (\ref{calA}), and with the quantities $\chi_j$ in $\cD$ arbitrary at this stage. We claim that $\tau_n$ is equal to the determinant $\tilde{\tau}_n$ of the matrix
\beq
\cM_n= \cC_{\cR}+\cB \tilde{\cA}_n\cB,
\eeq
where
\beq\label{calCR}
\cC_{\cR}=  \left( \begin{array}{cc}
0 & \cR_N \\
-\cR_N & 0 
\end{array} \right),
\eeq
\beq
\cB= \diag (\beta_1,\ldots,\beta_M),
\eeq
\beq
\tilde{\cA}_{n,jk}= \frac{(-a_j)^na_k^{-n+1}}{a_j+a_k},\ \ \ \ j,k=1,\ldots,M,
\eeq
provided that $\beta_1,\ldots,\beta_M$ are chosen such that
\beq\label{bcon}
\beta_j\beta_{M-j+1}=\chi_j^2,\ \ \ \ j=1,\ldots,N.
\eeq
(Recall $\cR_N$ denotes the $N\times N$ reversal permutation matrix.)

In order to prove this, we denote the columns of $\cM_n$ by $c_1,\ldots,c_M$, and use
\beq
\tilde{\tau}_n=|\cM_n|=|{\rm Col}(c_1,\ldots,c_M)|=|{\rm Col}(c_M,\ldots,c_{N+1},-c_N,\ldots,-c_1)|.
\eeq
Thus, $\tilde{\tau}_n$ equals the determinant of the matrix
\beq
{\bf 1}_M+\cN_n\cJ,\ \ \ \cN_{n,jk}= \beta_j\beta_{M-k+1}\frac{(-a_j)^n(a_{M-k+1})^{-n+1}}{a_j+a_{M-k+1}}.
\eeq
Transforming this matrix with the similarity matrix
\beq
\cS =\diag (\gamma_1,\ldots,\gamma_M),\ \ \ \gamma_j=(\beta_{M-j+1}/\beta_j)^{1/2},\ \ \ j=1,\ldots,M,
\eeq
we obtain
\beq
\tilde{\tau}_n=|{\bf 1}_M+\tilde{\cD}\cA_n\tilde{\cD}\cJ|,
\eeq
with
\beq
\tilde{\cD}= \diag ((\beta_1\beta_M)^{1/2},(\beta_2\beta_{M-1})^{1/2},\ldots,(\beta_M\beta_1)^{1/2}).
\eeq
Hence the $\beta$-constraint (\ref{bcon}) entails equality of $\tilde{\tau}_n$ and $\tau_n$, as advertised.

We can now compare the new representation
\beq\label{repn}
\tau_n=|\cC_{\cR}+\cD\tilde{\cA}_n\cD|,
\eeq
obtained by choosing
\beq
\beta_j=\beta_{M-j+1}=\chi_j,\ \ \ \ j=1,\ldots,N,
\eeq
(so that (\ref{bcon}) is obeyed, and $\cB$ and $\tilde{\cD}$ both equal $\cD$), with the solitons obtained by a $B_{\infty}$ symmetry reduction from the Darboux type 2D Toda solitons. (See Appendix~C for a sketch of their construction.) The reduced solitons are of the form
\beq\label{tau D}
\tau_n^D=|\cC+\tilde{\cB}\tilde{\cA}_n\tilde{\cB}|,
\eeq
where $\cC$ is an arbitrary antisymmetric $M\times M$ matrix and $\tilde{\cB}$ is a diagonal matrix
\beq
\tilde{\cB}=\diag (\tilde{\beta}_1,\ldots,\tilde{\beta}_M),
\eeq
with diagonal elements of the form
\beq\label{bD}
\tilde{\beta}_j=\alpha_j\exp(-iva_j+iua_j^{-1}),\ \ \ \ j=1,\ldots,M.
\eeq
Recalling the definition (\ref{chi})--(\ref{psi}) of $\chi_j$, we see that (\ref{bcon}) is obeyed when we choose $\alpha_1,\ldots,\alpha_M$ such that
\beq
\alpha_j\alpha_{M-j+1}=\exp(\varphi _j),\ \ \ \ j=1,\ldots,N.
\eeq
Provided we also specialise the arbitrary antisymmetric matrix $\cC$ to $\cC_{\cR}$ (cf.~(\ref{calCR})), we therefore conclude equality of the reduced Darboux type solitons to the reduced Kyoto solitons. 

We continue by using the new representation (\ref{repn}) to show the square property of $\tau_0$, cf.~(\ref{tsq}). As it stands, the matrix $\cC_{\cR}+\cD\tilde{\cA}_0\cD$ is not antisymmetric. But we have
\beq
\tilde{\cA}_{0,jk}=\frac{a_k}{a_j+a_k}=\frac{1}{2}\left(1-\frac{a_j-a_k}{a_j+a_k}\right),
\eeq
so that we can write
\beq
\tilde{\cA}_0=\frac{1}{2}\zeta\otimes\zeta -\frac{1}{2}\cA,\ \ \ \ \ \zeta= (1,\ldots,1),
\eeq
where $\cA$ is the antisymmetric matrix with elements
\beq
\cA_{jk}=\frac{a_j-a_k}{a_j+a_k}.
\eeq
As shown in Appendix~B, the determinant
\beq\label{tdet}
\tau_0=|\cC_{\cR}+\frac{1}{2}(\cD\zeta\otimes\cD\zeta-\cD\cA\cD)|
\eeq
equals the square of
\beq\label{tpf}
\tau ={\rm Pf}(\cC_{\cR}-\cD\cA\cD/2),
\eeq
and the pfaffian can be explicitly evaluated. The resulting formula is
\beq\label{KS}
\tau =\sum_{l=0}^N\sum_{1\le j_1<\cdots<j_l\le N}
\prod_{m=1}^l d_{j_m}\cdot \prod_{1\le m<n\le l}c_{j_mj_n},
\eeq
with
\beq
d_j=-\frac{1}{2}\chi_j^2\cA_{j,M-j+1},\ \ \ j=1,\ldots,N,
\eeq
\beq\label{cij}
c_{ij}=\cA_{ij}\cA_{i,M-j+1}\cA_{j,M-i+1}\cA_{M-j+1,M-i+1},\ \ \ \ i,j=1,\ldots,N,
\eeq
cf.~(\ref{KSalt}).

Next, we specialize $a_1,\ldots,a_M$ as in (\ref{as1})--(\ref{as2}). Using (\ref{chirep}), this yields
\beq
d_j=\frac{\cos(c)}{2i\sin(c)}\exp(\varphi_j-2\sin(c)[v\kappa_j+u\kappa_j^{-1}])
\eeq
\beq
c_{ij}=\left( \frac{\kappa_i-\kappa_j}{\kappa_i+\kappa_j}\right)^2\frac{\kappa_i^2+\kappa_j^2+2\cos(2c)\kappa_i\kappa_j}{\kappa_i^2+\kappa_j^2-2\cos(2c)\kappa_i\kappa_j},
\eeq
where
\beq
\kappa_j= \exp(\theta_j),\ \ \ \ j=1,\ldots,N.
\eeq
We now substitute (\ref{vphi}), and then choose $c=\pi/3$ and
  introduce new parameters
\beq
k_j=\sqrt{3}\kappa_j,\ \ \ \exp(s_j)=\exp(q_j/2)F_j(\theta)/2,\ \ \ \ j=1,\ldots,N,
\eeq
\beq
x=-v,\ \ \ \ y=-u.
\eeq
Then we finally obtain
\beq\label{dj}
d_j=\exp(s_j+xk_j+3yk_j^{-1}),
\eeq
\beq\label{ssc}
c_{ij}=\left( \frac{k_i-k_j}{k_i+k_j}\right)^2\frac{k_i^2+k_j^2-k_ik_j}{k_i^2+k_j^2+k_ik_j}.
\eeq

With $(q,\theta)$ varying over $\cP$ (\ref{cPT}), we have $k_j>0$ and $\exp(s_j)>0$, so that $d_j>0$ and $c_{ij}>0$. The function $\tau$ is therefore positive, implying $\tau_0$ is positive. Note that $\tau$ can be rewritten in the form (\ref{taun}), with $M$, $\xi_{j,n}$ and $f_{jk}$ replaced by $N$, $d_j$ and $c_{jk}$, respectively.

Last but not least, the tau-function (\ref{KS}) with the substitutions (\ref{dj})--(\ref{ssc}) coincides with the tau-function obtained in Section~2 of~\cite{kash}. Moreover, combining the relations
\beq
\exp(\Psi)=\tau_2/\tau_0,\ \ \ \tau_0=\tau_1,\ \ \ \ \tau_0=\tau^2,
\eeq
(cf.~(\ref{Psi}), (\ref{symt}) and (\ref{tsq})), and the 2D Toda equation of motion (\ref{Totau}) with $n=1$, we deduce
\beq
2\partial_u\partial_v \ln \tau =1-\exp(\Psi).
\eeq
Therefore, the function $\exp(\Psi)$ coincides with the function $v$ employed in~\cite{kash}.



\renewcommand{\theequation}{\thesection.\arabic{equation}}

\setcounter{equation}{0}

\section{Concluding remarks}

(i) ({\em Integrability on $\Omega_P$})
It is not hard to see that the flows generated by the Hamiltonians
\beq\label{Hr}
H_r={\rm Tr} (L^r),\ \ \ r\in\R^{*},
\eeq
leave $\Omega_P$ invariant, provided
\beq
\cos(rc)=\cos(c).
\eeq
Indeed, this readily follows from (\ref{flows}) and (\ref{xpr}), noting that (\ref{Hr}) corresponds to (\ref{Hf}) with $h(z)=e^{rz}$. Thus the restricted phase space $\langle \Omega_P,\sum_{j=1}^Ndq_j\wedge d\theta_j\rangle$ and the Hamiltonians
\beq
H_r,\ \ \pm r=1+2\pi k/c,\ \ \ \ k\in\Z,
\eeq
give rise to an integrable system on $\Omega_P$. For the Tzitzeica case $c=\pi/3$, one can choose as the $N$ independent Hamiltonians the power traces
\beq
{\rm Tr} (L^{1+6k}),\ \ \ \ k=0,1,\ldots,N-1.
\eeq
\vspace{4mm}

\noindent
(ii) ({\em Space-time trajectories})
As we have shown in Theorem~4.1, on $\Omega_P$ the tau-function \eqref{tauL} equals the Tzitzeica tau-function, and hence is real, cf.~Proposition~3.1. This reality property is far from obvious for $n=0,1$, but for $n=2$ reality on all of $\Omega$ is in fact clear from real-valuedness of $A$, cf.~\eqref{A1}. Specializing  to $\Omega_P$, we have
\beq
\tau_2(u,v)=\prod_{i=1}^{2N}[1-\exp(x_i^{+}(u,v))]\cdot\prod_{j=1}^N[1+\exp(x_j^{-}(u,v))].
\eeq
Now it follows from Section~6 that $\tau_0(u,v)$ is positive. Therefore, the Tzitzeica solution
\beq
\Psi(u,v)=\ln (\tau_2(u,v)/\tau_0(u,v)),
\eeq
has logarithmic singularities at the zeros of $\tau_2$, as we have already seen for $N=1$ in the previous section. We proceed to analyze these in terms of the space-time coordinates $t$ and $y$, cf.~(\ref{ty}). 

Fixing $t$, the function $\tau_2$ is positive for $y\to\pm\infty$, and has $2N$ sign changes for finite~$y$. The locations of these zeros on the $y$-axis are distinct for all $t$ (since $x_{2N}^{+}<\cdots<x_1^{+}$). Thus one obtains $2N$ space-time trajectories, which may be viewed as the locations of the particles in the $N$-soliton solution. For $t\to\pm \infty$ these trajectories exhibit soliton scattering, with a factorized phase shift in terms of the function $\ln (c_{ij})$, cf.~(\ref{ssc}). A more systematic analysis would be feasible by following the path laid out in Chapter~7 of~\cite{aa2}, but this is beyond our present scope. See, however, Fig.~2 for a plot of the 2-soliton collision. 
\begin{figure}[h]
	\centering
	\includegraphics*[width=6in]{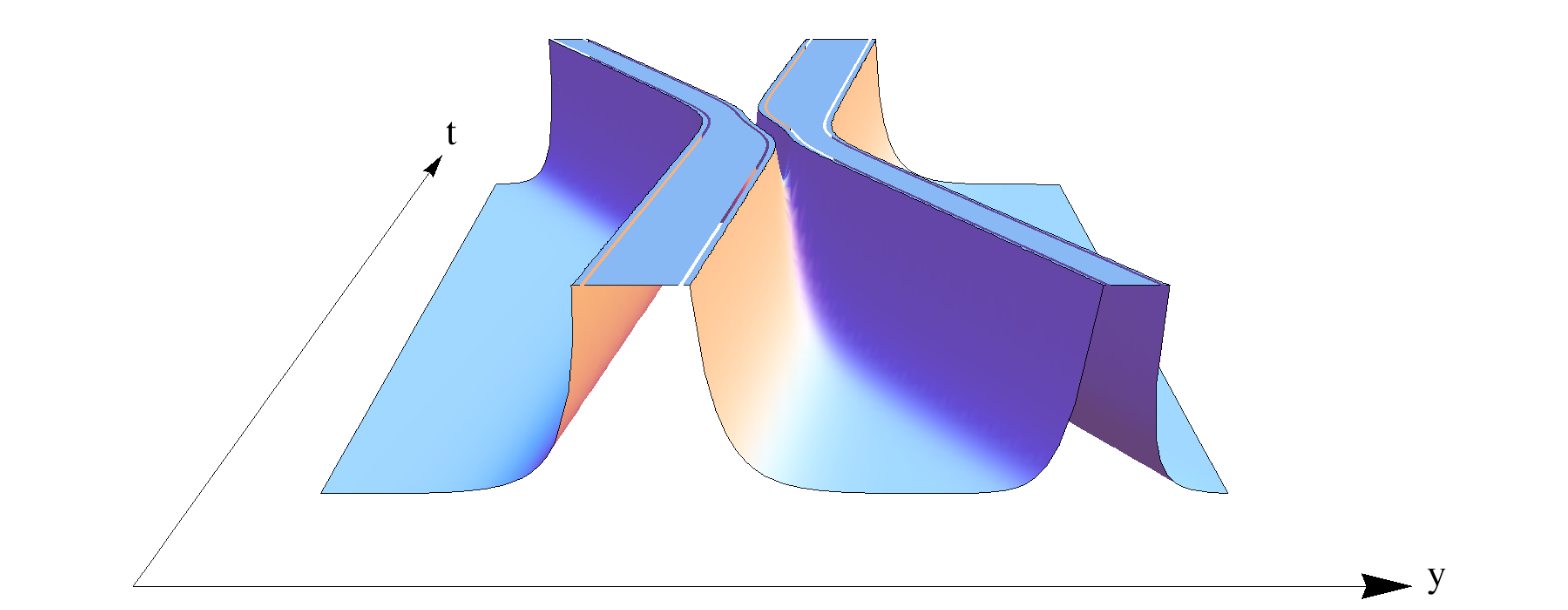}
	\caption{Two-soliton interaction. $-\Psi(t,y)$ is shown, the $y$-axis points left to right and the $t$-axis points into the page. The `plateau' approximates the region in which $\Psi(t,y)$ is complex-valued.}
\end{figure}
\vspace{4mm}

\noindent
(iii) (\emph{Comparison with} \cite{KP})
The 2D Toda soliton tau-functions studied in~\cite{KP} do not include the above real-valued Tzitzeica tau-functions. This is because in~\cite{KP} the condition $\tau_{-n}=\overline{\tau}_n$ is imposed (cf.~(2.10) in~\cite{KP}), whereas for the Tzitzeica case we have $\tau_{-1}=\tau_2\ne\tau_1=\overline{\tau}_1$. 
To accommodate this different starting point, the fusion procedure in Section~2 of~\cite{KP} starts from tau-functions that in terms of the relativistic Calogero-Moser systems amount to
\beq
\tau_n=|{\bf 1}_M +\exp(-2inc)A(x(u,v))|,\ \ \ A(x)=\diag(\exp(x_1^{+}),\ldots,\exp(x_M^{+})).
\eeq
Even so, we could specialize the fusion identities of~\cite{KP}, since they only pertain to the action variables, and the dependence on the latter is governed by the same function ((2.15) in~\cite{KP}) for particles and antiparticles. 
\vspace{4mm}

\noindent
(iv) (\emph{Quantum analogs}) We have chosen $N_{+}=2N$ and $N_{-}=N$ throughout the paper (cf.~(\ref{NN})), but we could just as well have started from $N$ particles and $2N$ antiparticles. Indeed, this amounts to working with the `charge conjugate' dual Lax matrix
\beq
A^C(x)=\diag (\exp(x_1^{+}),\ldots,\exp(x_N^{+}), -\exp(x_1^{-}),\ldots,-\exp(x_{2N}^{-})),
\eeq
and tau-function
\beq
\tau_n^C(u,v)= \det ({\bf 1}_{3N}-\exp(i\pi (1-2n)/3)A^C(x(u,v))),
\eeq
instead of (\ref{A1}) and (\ref{tauL}). Clearly, one can then proceed in the same way as before.

On the other hand, one cannot enlarge the 3-body--soliton correspondence without venturing into complexified phase spaces, losing control of the action-angle map in the process. More precisely, there exist real-valued Tzitzeica soliton tau-functions that correspond to two different charges, hence giving rise to breather-like bound states. (Indeed, this is not hard to see from the explicit form of the solitons at the end of Section~3; for example, one need only perform a suitable analytic continuation in the 2-soliton solution to obtain the 1-breather solution.) However, it can be shown that these do not correspond to subspaces of the real Calogero-Moser phase spaces with arbitrary $N_{+}$ and $N_{-}$. 

It may be expected that this picture persists on the quantum level. To be specific, a suitable reduction of the unitary joint eigenfunction transform for the commuting quantum Hamiltonians with $c=\pi/3$ and $3N$ variables (which, to be sure, is not even known to exist to date) should give rise to a unitary transform with $N$ variables, which can be interpreted as a transform for $N$ quantum solitons with the same charge. However, no such reductions are likely to exist when different charges are involved. In particular, if quantum breathers for the Tzitzeica quantum field theory do exist (a feature that is taken for granted in most of the work within the form factor program), then they are not likely to have analogs in the Calogero-Moser particle picture. (By contrast, for the sine-Gordon case a complete correspondence is expected~\cite{Kiev}.)
\vspace{4mm}

\noindent
(v) ({\em Demoulin solitons})
From Propositions~2.2 and~2.3 it is clear that for 
\beq
c=\pi/6
\eeq
the tau-function satisfies
\beq
\tau_{-n+1}=\tau_n,\ \ \ \ \tau_{n+6}=\tau_n.
\eeq
As a consequence, we have
\beq\label{ttt}
\tau_1=\tau_0,\ \ \tau_2=\tau_5,\ \ \tau_3=\tau_4.
\eeq
Setting
\beq
h_n=\tau_{n+1}\tau_{n-1}/\tau_n^2=\exp(\phi_n-\phi_{n-1}), 
\eeq
(where we used (\ref{phitau})), we deduce
\beq
h_1=h_0,\ \ h_2=h_5,\ \ h_3=h_4,
\eeq
and
\beq\label{hpr}
h_1h_2h_3=1.
\eeq
Now from (\ref{To}) we have
\beq
(\ln h_n)_{uv}=2h_n-h_{n+1}-h_{n-1}.
\eeq
Hence, setting
\beq
h=h_1=\tau_2/\tau_1,\ \ \ k=h_3=\tau_2/\tau_3,
\eeq
and using (\ref{hpr}), we obtain
\beq
(\ln h)_{uv}=h-\frac{1}{hk},\ \   (\ln k)_{uv}=k-\frac{1}{kh}.
\eeq

This system of relativistic wave equations is the Demoulin system, cf.~\cite{rosc2}, p.~343, Eq.~(9.58). Therefore the $c=\pi/6$ tau-functions with parameters in $\cP$ yield real-valued Demoulin solitons in tau-function form. In particular, from (\ref{quv})--(\ref{eq:pp1}) we see that the one-soliton case is given by (\ref{ttt}) and
\beq
\tau_1=1+2G+G^2,\ \ \tau_2=1+G^2,\ \ \tau_3=1-2G+G^2,\ \  G=\exp(q/2-ve^{\theta}-ue^{-\theta}).
\eeq

\appendix
																																												
\section*{Appendices}


\section{Fusion}          

In this appendix we detail how the formula \eqref{pI} for $p_I(\hat\theta)^2$  leads to  \eqref{limI}
in the collapsing torus limit. We begin by recalling that in this limit $p_I(\hat\theta)$ vanishes, unless the subset $I$ of $\{1,\dots,3N\}$ has the form
\beq
I=\bigcup_{j\in J}I_j                   
\eeq
where 
\beq
J=I\cap\{1,\dots,2N\},
\ \ \ 
I_j=\{j\},\ \ I_{N+j}=\{N+j,2N+j\},\ \ j=1,\dots,N,
\eeq
 cf.~the paragraph preceding \eqref{If}.  For a subset $I$ of this form, \eqref{pI} entails
\begin{equation}\label{eq:pI}
p_I(\hat\theta)^2=\prod_{j\in J,k\not\in J}\left(
\prod_{m\in I_j,n\in I_k}\frac{\sinh^2((\hat\theta_m-\hat\theta_n)/2)+\sin^2c}{\sinh^2((\hat\theta_m-\hat\theta_n)/2)}\right).
\end{equation}
 
Next, we substitute \eqref{qhq0} and \eqref{qhq1} in this formula and cancel terms in the interior product using fusion identities. A general fusion identity, in which the sets $I_k$ are of arbitrary cardinality, can be found in \cite{KP}, but here we only detail the special cases we need.

Using the trigonometric/hyperbolic identity 
\beq
\sinh^2x+\sin^2y=\sinh(x+iy)\sinh(x-iy),
\eeq
the general pair factor in \eqref{eq:pI} can be written as
\beq
\frac{\sinh((\hat\theta_m-\hat\theta_n)/2+ic)\sinh((\hat\theta_m-\hat\theta_n)/2-ic)}{\sinh^2((\hat\theta_m-\hat\theta_n)/2)}.
\eeq                                
For a soliton-soliton interaction included in $J$, there is only one pair factor in the product. Specifically,  we have $I_j=\{j\}$, $I_{k}=\{k\}$ and (from \eqref{qhq0} and \eqref{eq:eta})  $\hat\theta_j=\eta_j,\ \hat\theta_k=\eta_k$, so the contribution to the product for a soliton-soliton pair is
\beq
\frac{\sinh^2((\eta_j-\eta_k)/2)+\sin^2c}{\sinh^2((\eta_j-\eta_k)/2)}=
\frac{\sinh((\eta_j-\eta_k)/2+ic)\sinh((\eta_j-\eta_k)/2-ic)}{\sinh^2((\eta_j-\eta_k)/2)}=
\frac{s(2)s(-2)}{s^2(0)},
\eeq                                                                  
where we have set $s(l)=\sinh((\eta_j-\eta_k+lic)/2)$.
For a soliton-breather interaction we have $I_j=\{j\}$, $I_{N+k}=\{N+k,2N+k\}$ and $\hat\theta_j=\eta_j$, $\hat\theta_{N+k}=\eta_k+ic$, $\hat\theta_{2N+k}=\eta_k-ic$, yielding two pair factors                                                                 
\beq
\frac{s(1)s(-3)}{s^2(-1)}
\frac{s(3)s(-1)}{s^2(1)}=
\frac{s(3)s(-3)}{s(1)s(-1)}=
\frac{\sinh^2((\eta_j-\eta_k)/2)+\sin^2(3c/2)}{\sinh^2((\eta_j-\eta_k)/2+\sin^2(c/2)}.
\eeq          
Finally, for a breather-breather interaction in $J$, we get $I_{N+j}=\{N+j,2N+j\}$, $I_{N+k}=\{N+k,2N+k\}$, and hence four pairs whose contribution to the product is
\beq
\frac{s(2)s(-2)}{s^2(0)}
\frac{s(4)s(0)}{s^2(2)}
\frac{s(2)s(-2)}{s^2(0)}
\frac{s(0)s(-4)}{s^2(-2)}
=
\frac{s(4)s(-4)}{s^2(0)}=
\frac{\sinh^2((\eta_j-\eta_k)/2)+\sin^2(2c)}{\sinh^2((\eta_j-\eta_k)/2)}.
\eeq         
Using $\eta_{l+N}=\eta_l$ and $c_l=c$, $c_{N+l}=2c$ for $l=1,\dots,N$ (cf.~\eqref{eq:eta}), these  three cases may be 
encoded in the single formula
\beq
\frac{\sinh^2((\eta_j-\eta_k)/2)+\sin^2((c_j+c_k)/2)}{\sinh^2((\eta_j-\eta_k)/2+\sin^2((c_j-c_k)/2)},\ \ \ j,k=1,\ldots,2N.
\eeq 
Hence \eqref{limI} results.        


\section{Pfaffian identities}

In this appendix we show that the determinant on the rhs of (\ref{tdet}) equals the square of the pfaffian on the rhs of (\ref{tpf}), and that the latter has the explicit form (\ref{KS})--(\ref{cij}). We proceed in a slightly more general way, since this eases the notation and adds insight.

First we note that the determinant is of the form $|v\otimes v+A|$, where $A$ is a $2N\times 2N$ antisymmetric matrix. Now we invoke the following lemma.

\begin{lem}
Assume $A$ is an antisymmetric $2N\times 2N$ matrix and $v\in\C^{2N}$. Then one has
\beq\label{id}
|v\otimes v+A|=|A|.
\eeq
\end{lem}
\begin{proof} By continuity it suffices to prove this for an invertible $A$. Then we have
\beq
|v\otimes v+A|=|A| |{\bf 1}_{2N}+A^{-1}v\otimes v|=|A|(1+(\overline{v},A^{-1}v)).
\eeq
Now since $(A^{-1})^t=(A^t)^{-1}=-A^{-1}$, it follows that $A^{-1}$ is also antisymmetric. Hence the inner product $(\overline{v},A^{-1}v)=v^tA^{-1}v$ vanishes, and (\ref{id}) results.
\qedhere
\end{proof}

As a consequence, we have
\beq
\tau_0=|\cC_{\cR}-\cD\cA\cD/2|=\tau^2,
\eeq
with $\tau$ given by (\ref{tpf}). We continue to show (\ref{KS}). Again, we first study a slightly more general situation. We begin by recalling the definition
\beq\label{Pf}
{\rm Pf}(A)=\sum (-)^{{\rm sgn}(\sigma)} A_{i_1j_1}\cdots A_{i_Nj_N},\ \ \ A\in M_{2N}(\C),\ \ A^t=-A,
\eeq
where the sum is over all index choices with
\beq
1=i_1<\cdots<i_N\le 2N-1,\ \ \ \ i_1<j_1,\ldots , i_N<j_N,
\eeq
and $\sigma$ is the permutation
\beq
\sigma\, :\, (i_1,j_1,\ldots,i_N,j_N)\mapsto (1,2,\ldots,2N).
\eeq
Denoting ${\rm Pf}(A)$ by $[1,2,\ldots,2N]$, this implies an expansion formula
\beq\label{ex}
[1,2,\ldots,2N]=\sum_{n=2}^{2N}(-)^n[1,n][2,\ldots,\hat{n},\ldots,2N],
\eeq
where the  notation on the rhs will be clear from context. Consider now the pfaffian of the antisymmetric matrix
\beq
B=
\left(\begin{array}{cc}
0 & \lambda\cR_N \\
-\lambda\cR_N & 0 
\end{array} \right)+A,\ \ \ \ \lambda\in\C.
\eeq
The relevant (upper triangular) part of $B$ is
\beq
\begin{array}{ccccccc}
A_{12} & \cdot  & \cdot& \cdot& \cdot& \cdot & \lambda+A_{1,2N}\\
  & \cdot & \cdot& \cdot& \cdot& \cdot& \cdot\\
  &  & A_{N-1,N} & A_{N-1,N+1}  &  \lambda +A_{N-1,N+2}    & \cdot & \cdot\\
  &  &  & \lambda +A_{N,N+1}  &  A_{N, N+2}  & \cdot  & \cdot\\
   &  &  &  &  A_{N+1,N+2}  & \cdot & \cdot \\
    &  &  &  &  & \cdot & A_{2N-2,2N} \\
     &  &  &  &  &  &  A_{2N-1,2N}\\
\end{array} 
     \eeq
Therefore we have
\beq
{\rm Pf}(B)=\sum_{l=0}^N\lambda^{N-l}S_l,
\eeq
where $S_l$ is a sum over certain `minors'. Specifically, taking (\ref{ex}) into account to check signs, we readily get
\beq
S_k=\sum_{1\le j_1<\cdots<j_k\le N}[j_1,\ldots,j_k,2N-j_k+1,\ldots,2N-j_1+1].
\eeq
(Hence in particular $S_0=1, S_{2N}=[1,\ldots,2N]$.)

Applying this expansion to $\tau$ (\ref{tpf}), we obtain
\beq\label{pfA}
\tau =\sum_{l=0}^N\sum_{1\le j_1<\cdots<j_l\le N}\prod_{m=1}^l\left( -\frac{1}{2}\chi_{j_m}^2\right)\cdot [j_1,\ldots,j_l,2N-j_l+1,\ldots, 2N-j_1+1]_{\cA},
\eeq
where $[...]_{\cA}$ denotes the reduced pfaffians for $\cA$. We now obtain an explicit formula for the latter.

\begin{lem}
Suppose $\cA\in M_{2l}(\C)$ is given by
\beq
\cA_{jk}=\frac{a_j-a_k}{a_j+a_k},\ \ \ \ j,k=1,\ldots,2l.
\eeq
Then we have
\beq\label{Aid}
{\rm Pf}(\cA)=\prod_{1\le j<k\le 2l}\cA_{jk}.
\eeq
\end{lem}


This explicit evaluation is proved in~\cite{NimHL}. It is presumably known, but we do not know a reference. Indeed, up to its sign this lemma can be viewed as a consequence of Cauchy's identity. To see this, note that the lemma entails
\beq\label{Cau}
|\cA|=\prod_{1\le j<k\le 2l}\cA_{jk}^2.
\eeq
If we now substitute $a_j\to \exp(2q_j)$, then (\ref{Cau}) becomes
\beq
|(\tanh (q_j-q_k))|_{2l\times 2l}=\prod_{1\le j<k\le 2l}\tanh^2(q_j-q_k),
\eeq
an identity that readily follows from Cauchy's identity (cf.~(B42) in~\cite{comp}). Thus, Cauchy's identity implies (\ref{Cau}), so (\ref{Aid}) follows save for its sign.

Returning to (\ref{pfA}), we substitute (\ref{Aid}) in the $\cA$-pfaffians, yielding
\beq\label{pfeq}
\prod_{j<k,j,k\in \{j_1,\ldots,j_l,2N-j_l+1,\ldots,2N-j_1+1\}}\cA_{jk}.
\eeq
Introducing
\beq
c_{jk}= \cA_{jk}\cA_{j,2N-k+1}\cA_{k,2N-j+1}\cA_{2N-k+1,2N-j+1},\ \ \ 1\le j<k\le N,
\eeq
we can write (\ref{pfeq}) as
\beq
\prod_{m=1}^l\cA_{j_m,2N-j_m+1}\cdot \prod_{j<k,j,k\in\{ j_1,\ldots,j_l\}} c_{jk}.
\eeq
Substituting this in (\ref{pfA}) and setting
\beq
d_j=-\frac{1}{2}\chi_j^2\cA_{j,2N-j+1},\ \ \ \ j=1,\ldots,N,
\eeq
we finally obtain
\beq\label{KSalt}
\tau=\sum_{l=0}^N\sum_{1\le j_1<\cdots <j_l\le N}\prod_{m=1}^l d_{j_m}\cdot
\prod_{j<k,j,k\in\{ j_1,\ldots,j_l\}}c_{jk},
\eeq
which amounts to (\ref{KS}).   


\newcommand{\Pf}{\mathrm{Pf}}
\section{Solutions obtained by binary Darboux transformation}

The $A_\infty$ Toda lattice (1.3) has a family of solutions obtained by an iterated binary Darboux transformation. We follow the derivation given in \cite{NimWil}, but make some small modifications so as to conform with the notation and conventions used in the rest of this paper. 

A Lax pair for \eqref{To} is given by
\beq
\label{lax}	
A_{n,u}=i\exp(\phi_{n+1}-\phi_{n})A_{n+1},\quad 
-A_{n,v}=iA_{n-1}+\phi_{n,v}A_n.
\eeq
(Thus, equality of cross derivatives entails (\ref{To}).)
It is covariant with respect to the Darboux transformation \cite{Mat}
\beq
\tilde A_n=A_{n,v}-a_{n,v}a_n^{-1}A_n,
\eeq
\beq \tilde\phi_n=\phi_n+\ln(a_{n}/a_{n-1}),
\eeq
where $a_n$ is an eigenfunction (i.e., any particular solution of \eqref{lax}). The adjoint Lax pair is 
\beq
\label{adlax}	
-B_{n,u}=i\exp(\phi_{n}-\phi_{n-1})B_{n-1},\quad
B_{n,v}=iB_{n+1}+\phi_{n,v}B_n.
\eeq

Given an eigenfunction $a_n$ and an adjoint eigenfunction $b_n$ for a seed solution $\phi_n=\phi_n^0=\ln(\tau^0_{n+1}/\tau^0_{n})$ to (1.3), an eigenfunction potential $\Omega(a_n,b_n)$ is defined, consistently and uniquely up to an additive constant, by the three requirements
\beq
\Omega(a_n,b_n)_{,u}    =   -i\exp(\phi^0_{n+1}-\phi^0_{n})b_{n}a_{n+1},
 \eeq
 \beq
\Omega(a_n,b_n)_{,v}   =    -ib_{n+1}a_{n},
\eeq
\beq
\Omega(a_n,b_n)-\Omega(a_n,b_{n-1})   =    -b_{n}a_{n}.
\eeq
The Lax pair \eqref{lax} is also covariant with respect to the binary Darboux transformation
\beq\label{cov}
\hat A_n=A_n-a_{n}\Omega(A_n,b_n)/\Omega(a_n,b_n),
\eeq
\beq
\hat\phi_n=\phi_n^0+\ln(\Omega(a_{n},b_{n})/\Omega(a_{n-1},b_{n-1})),
\eeq
or, equivalently, $\hat\tau_n=\Omega(a_{n-1},b_{n-1})\tau^0_n$. (In (\ref{cov}) $A_n$ denotes the general solution to (\ref{lax}) with $\phi$ replaced by $\phi^0$.)

The $M$th iterated binary Darboux transformation is expressed in terms of an $M$-vector  $\mathbf A_n$, whose components satisfy \eqref{lax} with $\phi \to \phi^0$, and an $M$-vector $\mathbf B_n$, whose components satisfy \eqref{adlax} with $\phi\to\phi^0$. These eigenfunction and adjoint eigenfunction vectors are used to define an $M\times M$ matrix eigenfunction potential $\Omega_n=\Omega(\mathbf A_n,\mathbf B_n)$ satisfying
\beq\label{ef pot}
\Omega_{n,u}=-i\exp(\phi^0_{n+1}-\phi^0_{n})\mathbf B_{n}\otimes\mathbf A_{n+1},\quad
\Omega_{n,v}=-i\textbf B_{n+1}\otimes\mathbf A_{n},\quad
\Omega_{n}-\Omega_{n-1}=-\mathbf B_{n}\otimes\mathbf A_{n}.
\eeq
The solutions obtained are then expressed as
\beq
	    \hat\tau_{n}=\det(\Omega_{n-1})\tau^0_n.
\eeq
From now on we omit the $\hat\ $ and write $\hat\tau_n$ as $\tau_n$.

The $B_\infty$ Toda lattice has $\phi_{-n}=-\phi_{n}$ for $n\in\mathbb Z$, so that $\phi_0=0$ and the system is semi-infinite, expressed in terms of $\phi_1,\phi_2,\dots$. This reduction can be achieved by choosing solutions satisfying $\tau_{-n+1}=\tau_n$. In this reduction, it is consistent to choose adjoint eigenfunctions given by $\mathbf B_{n}=(-1)^{n}\mathbf A_{-n}$. It then follows from \eqref{ef pot} that
\begin{equation}\label{symm}
	\Omega_{-n}=-\Omega_{n-1}^t
\end{equation}
and so provided $M$ is even, $M=2N$ say, we indeed obtain $\tau_{-n+1}=\tau_{n}$.

In particular, to obtain soliton solutions for the $A_\infty$ Toda lattice, one chooses as seed solution the vacuum solution $\tau^0_n=1, \phi_n^0=0$, and then vectors of eigenfunctions and adjoint eigenfunctions given by
\begin{equation}\label{eq: Phi,Psi}
    \mathbf A_{n,j}=\alpha_j a_j^{-n}e^{-iva_j+iua_j^{-1}},\quad
    \mathbf B_{n,j}=\beta_j b_j^ne^{ivb_j-iub_j^{-1}},\ \ \ \ j=1,\ldots,M,
\end{equation}
where $\alpha_j,\beta_j,a_j$ and $b_j$ are arbitrary constants. Then, one obtains
\beq	    \tau_n=\det\left(C_{jk}+\beta_j\alpha_k\frac{b_j^{n}a_k^{-n+1}}{a_k-b_j}e^{-iv(a_k-b_j)+iu(a_k^{-1}-b_j^{-1})}\right),
\eeq
where $C_{jk}$ are arbitrary constants. In the $B_\infty$-reduction, one must take $b_i=-a_i$ and $\beta_i=\alpha_i$ giving 
\beq	    \tau_n=\det\left(C_{jk}+\alpha_j\alpha_k\frac{(-a_j)^{n}a_k^{-n+1}}{a_j+a_k}e^{-iv(a_j+a_k)+iu(a_j^{-1}+a_j^{-1})}\right).
\eeq
Here, $C_{jk}$ are arbitrary constants satisfying, in accordance with the symmetry condition \eqref{symm}, $C_{kj}=-C_{jk}$.

Thus we have arrived at the expression \eqref{tau D} that is the starting point of comparison with the Kyoto solitons in Section~6.

\section*{Acknowledgments}
We would like to thank E.~Ferapontov, E.~Corrigan, G.~Delius, C.~Korff, W.~Schief, M.~Niedermaier and C.~Rogers (in chronological order) for illuminating discussions and for making us aware of relevant literature. Also, thanks to suggestions of the referee the exposition of this paper has been improved.
The paper was completed while the authors were both visiting the Isaac Newton Institute, Cambridge, UK, as part of the programme on Discrete Integrable Systems (January-June 2009). They would like to thank the programme organisers for the invitations and the INI for financial support and hospitality.



\addcontentsline{toc}{section}{References}

\end{document}